\documentclass{article}

\usepackage{amsmath,amssymb,amsthm,MnSymbol}
\usepackage{url}
\usepackage{comment}
\usepackage{proof}
\usepackage{stmaryrd}
\usepackage{parskip}
\usepackage{fullpage}

\newcommand{\choice}[0]{\zeta}
\newcommand{\elcap}[0]{\cap}
\newcommand{\abs}[4]{{#1}\, #2\! : \! #3.\, #4}
\newcommand{\absu}[3]{{#1}\, #2.\, #3}
\mathchardef\mhyph="2D % Define a "math hyphen"
\newcommand{\interp}[1]{\llbracket #1 \rrbracket} 
\newcommand{\leadstoc}[0]{\ensuremath{\leadsto_{\mathbf{n}}}}
\newcommand{\leadstocs}[0]{\ensuremath{\leadsto_{\mathbf{n}}^*}}
\newcommand{\nleadstoc}[0]{\ensuremath{\nleadsto_{\mathbf{n}}}}
\newcommand{\tpcheck}[0]{\Leftarrow}
\newcommand{\tpsynth}[0]{\Rightarrow}
\newcommand{\tpsynthleads}[0]{\ensuremath{\overset{\leadstocs}{\Rightarrow}}}
\newcommand{\cbe}[0]{c\beta\eta}

\newcommand{\startcase}[1]{\vspace{#1} \noindent\textbf{\underline{Case:}}}

\newtheorem{theorem}{Theorem}
\newtheorem{definition}[theorem]{Definition}
\newtheorem{lemma}[theorem]{Lemma}
\newtheorem{corollary}[theorem]{Corollary}
\newtheorem{observation}[theorem]{Observation}

\begin{document}

\title{Syntax and Semantics of Cedille}

\author{Aaron Stump, Chris Jenkins \\
  Computer Science \\
  The University of Iowa \\
\texttt{aaron-stump@uiowa.edu}}

\date{}

\maketitle

\section{Introduction}

The type theory of Cedille is called the Calculus of Dependent Lambda
Eliminations (CDLE).  This document presents the version of CDLE as of
April 13, 2021.  We have made many changes from the first paper on CDLE~\cite{stump17},
mostly in the form of dropping constructs we discovered (to our
surprise) could be derived~\cite{stump18}. We have also omitted
\emph{lifting} -- a technique for large eliminations with lambda
encodings -- in this document's version of CDLE.  Some uses of lifting
can be simulated other ways within the system, though the limits of
this are still under investigation.  We also include a construct
$\delta$, for deriving a contradiction from a proof that
lambda-encoded true equals lambda-encoded false.  This also
compensates somewhat for the lack of lifting.

At a high level, CDLE is an extrinsic (i.e., Curry-style) type theory
extending the Calculus of Constructions with three additional
constructs, which allow deriving induction principles within the
theory for lambda encodings of inductive datatypes.  The goal is to support
usual idioms of dependently typed programming and proving as in Agda or
similar tools, but using pure lambda encodings for all data, and
requiring a much smaller core theory suitable for formal verification.

The current Cedille implementation of CDLE extends with a number of features
intended to make programming in the system more convenient and with less
redundancy.
These features all compile away to a slightly simplified version of the
theory presented in this document, called Cedille Core, described
here: \url{https://github.com/astump/cedille-core-spec}.
At the time of writing, the Cedille implementation lags behind the formulation
of CDLE in this document on one point: the treatment of the rewriting construct
\(\rho\).
In the implementation, the fully-annotated form for this construct is \(\rho\ t\
@x.T'\ \mhyph\ t'\), but in this document it is \(\rho\ t\ @x\langle
t_2 \rangle.T'\ \mhyph\ t'\).
The implementation infers the additional subterm \(t_2\).

\section{Classification Rules}

The classification rules are given in
Figures~\ref{fig:superknd},~\ref{fig:knd}, and~\ref{fig:tp}, with
Figure~\ref{fig:ctxt} giving the context formation rules.
For brevity, we take these figures as implicitly specifying the syntax of
contexts \(\Gamma\), kinds $\kappa$, types $T$, and annotated terms $t$; these
may use term variables $x$ and type variables $X$, which we assume come from
distinct sets.  So terms and types are syntactically distinguished.
We follow the syntax of our implementation
Cedille, which distinguishes application of a term or type $e$ to a
type ($e \cdot T$), from application to a term ($e\ t$), and
application to an erased term argument ($e\ \mhyph t$, in which case \(e\) must
be a term).
Note that center dot (\(\cdot\)) is also used to denote the empty type context;
since the usage for typing contexts always has the symbol occur to the left of
the turnstile (\(\vdash\)), no confusion should arise from overloading notation.

The typing rules (Figure~\ref{fig:tp}) are
bidirectional~\cite{pierce+00}, while the kinding rules (Figure~\ref{fig:knd})
are unidirectional (synthesizing only) and the kind-formation rules
(Figure~\ref{fig:superknd}) have no notion of directionality.
We write $\Leftrightarrow$ to range over
$\{\tpcheck,\tpsynth\}$, and when this symbol occurs multiple times in a rule,
it is intended that such occurrences be read the same way (i.e., read the
occurrences as either all \(\tpcheck\) or all \(\tpsynth\)).
The rules are intended to be read bottom-up as an algorithm (in a standard way,
c.f.~\cite{peytonjones07,Pfe01_Lecture-Notes-on-Bidirectional-Type-Checking})
for synthesizing a classifier from a context and an expression (\emph{type
  synthesis}, \(\tpsynth\)) or checking an expression against a classifier in a
context (\emph{type checking}, \(\tpcheck\)).
Since variables have their type synthesized, and since we sometimes wish to
substitute a variable with a term whose type can only be checked, we define a
shorthand notation: \([t/x]^T\) means \([\chi\ T\ \mhyph\ t/x]\), where \(\chi\)
is the construct for explicit type annotations (see Section~\ref{sec:overview}).

In typing rules, when the type of an introduction form is checked or the type of an
elimination form is synthesized, we use call-by-name normalization, written
\(\leadstoc\) and \(\leadstocs\) for the reflexive transitive closure, to put
types in weak head normal form, revealing type constructors.
We abbreviate the conjunction ``term \(t\) synthesizes some type, and that type
call-by-name reduces to another type'' with the symbol \(\tpsynthleads\),
defined formally at the top of Figure~\ref{fig:tp}.
When a redex (reducible expression) occurs such that the argument to a type is a
term, such as \((\abs{\lambda}{x}{T_1}{T_2})\ t_1\), the reduction uses
substitution with annotations: \([t_1/x]^{T_1}\ T_2\).

Call-by-name reduction is \emph{not} what underpins the
equivalence relation \(\cong\) for types and kinds, which is full
\(\beta\)-equivalence (for types) and \(\beta\eta\)-equivalence (for terms),
both of which are modulo erasure of annotations in terms.
The erasure operation is defined in Figure~\ref{fig:eraser}, and is essentially
the \emph{extraction} function for the Implicit Calculus of Constructions given
by~\cite{BB08_ICC-as-a-Programming-Language} adapted to CDLE.
To understand the role of erasure, recall that the type theory is \emph{extrinsic} (a.k.a.\
Curry-style), and hence we only consider erasures \(|t|\) of terms when testing
convertibility.
This is lifted to the conversion relation on types \(T \cong T'\) and kinds
\(\kappa \cong \kappa'\), whose rules are given in
Figure~\ref{fig:conv}.

\begin{figure}
  \[
  \begin{array}{ccc}
    \infer{\Gamma \vdash \star}{\ } &
    \infer{\Gamma\vdash\abs{\Pi}{x}{T}{\kappa}}{\Gamma \vdash T \tpsynth \star & \Gamma,x:T\vdash\kappa} &
    \infer{\Gamma\vdash\abs{\Pi}{X}{\kappa'}{\kappa}}{\Gamma \vdash \kappa' & \Gamma,X:\kappa'\vdash\kappa}
  \end{array}
  \]
  \caption{Rules for checking that a kind is well-formed ($\Gamma \vdash \kappa$)}
  \label{fig:superknd}
\end{figure}

\begin{figure}
  \[
    \begin{array}{c}
      [t/x]^T\ =\ [\chi\ T\ \mhyph\ t/x]
      \\ \\ \\
  \begin{array}{cc}
    \infer{\Gamma \vdash X \tpsynth \kappa}{(X : \kappa) \in \Gamma} &
    \infer{\Gamma\vdash \abs{\forall}{X}{\kappa}{T} \tpsynth \star}{\Gamma \vdash \kappa & \Gamma,X:\kappa\vdash T \tpsynth \star} \\ \\
    \infer{\Gamma\vdash\abs{\forall}{x}{T}{T'} \tpsynth \star}{\Gamma \vdash T \tpsynth \star & \Gamma,x:T\vdash T' \tpsynth \star} &
    \infer{\Gamma\vdash\abs{\Pi}{x}{T}{T'} \tpsynth \star}{\Gamma \vdash T \tpsynth \star & \Gamma,x:T\vdash T' \tpsynth \star} \\ \\
    \infer{\Gamma\vdash\abs{\lambda}{x}{T}{T'} \tpsynth \abs{\Pi}{x}{T}{\kappa}}{\Gamma \vdash T \tpsynth \star & \Gamma,x:T\vdash T'\tpsynth\kappa} &
    \infer{\Gamma\vdash\abs{\lambda}{X}{\kappa}{T'} \tpsynth \abs{\Pi}{X}{\kappa}{\kappa'}}{\Gamma \vdash \kappa & \Gamma,X:\kappa\vdash T'\tpsynth\kappa'} \\ \\
    \infer{\Gamma\vdash T\ t \tpsynth [t/x]^{T'}\kappa}{\Gamma\vdash T \tpsynth \abs{\Pi}{x}{T'}{\kappa} & \Gamma\vdash t \tpcheck T'} &
    \infer{\Gamma\vdash T_1 \cdot T_2 \tpsynth [T_2/X]\kappa_1}{\Gamma\vdash T_1 \tpsynth \abs{\Pi}{X}{\kappa_2}{\kappa_1} & \Gamma\vdash T_2 \tpsynth \kappa_2' & \kappa_2 \cong \kappa_2'} \\ \\
    \infer{\Gamma\vdash\abs{\iota}{x}{T}{T'} \tpsynth \star}{\Gamma \vdash T \tpsynth \star & \Gamma,x:T\vdash T' \tpsynth \star} &
    \infer{\Gamma\vdash \{ t \simeq t' \} \tpsynth \star}{\textit{FV}(t\ t')\subseteq\textit{dom}(\Gamma)}
  \end{array}
    \end{array}
  \]
  \caption{Rules for synthesizing a kind for a type ($\Gamma \vdash T \tpsynth \kappa$)}
  \label{fig:knd}
\end{figure}

\begin{figure}
  \[
    \begin{array}{c}
      \Gamma \vdash t \tpsynthleads T\ =\ \exists T'.\ (\Gamma
      \vdash t \tpsynth T') \land (T' \leadstocs T)
      \\ \\ \\
  \begin{array}{cc}
    \infer{\Gamma\vdash x\tpsynth T}{(x : T)\in\Gamma} &
    \infer{\Gamma\vdash t\tpcheck T}{\Gamma\vdash t\tpsynth T' & T' \cong T} \\ \\
    \infer{\Gamma\vdash \absu{\lambda}{x}{t} \tpcheck T}{T \leadstocs \abs{\Pi}{x}{T_1}{T_2} & \Gamma,x:T_1\vdash t\tpcheck T_2} &
    \infer{\Gamma\vdash t\ t' \tpsynth [t'/x]^{T'}T}{\Gamma\vdash t \tpsynthleads \abs{\Pi}{x}{T'}{T} & \Gamma\vdash t' \tpcheck T'} \\ \\

    \infer{\Gamma\vdash \absu{\Lambda}{X}{t} \tpcheck T'}
          {T' \leadstocs \abs{\forall}{X}{\kappa}{T} & \Gamma,X:\kappa\vdash t \tpcheck T} &
    \infer{\Gamma\vdash t \cdot T' \tpsynth [T'/X]T}
          {\Gamma\vdash t \tpsynthleads \abs{\forall}{X}{\kappa}{T} & \Gamma\vdash T' \tpsynth \kappa' & \kappa'\cong\kappa} \\ \\

    \infer{\Gamma\vdash \absu{\Lambda}{x}{t} \tpcheck T}
          {T \leadstocs \abs{\forall}{x}{T_1}{T_2} & \Gamma,x:T_1\vdash t \tpcheck T_2 & x\not\in\textit{FV}(|t|)} &
    \infer{\Gamma\vdash t\ \mhyph t' \tpsynth [t'/x]^{T'}T}{\Gamma\vdash t \tpsynthleads \abs{\forall}{x}{T'}{T} & \Gamma\vdash t' \tpcheck T'} \\ \\

    \infer{\Gamma\vdash [ t_1 , t_2 ] \tpcheck T}
          {
           \begin{array}{cc}
             T \leadstocs \abs{\iota}{x}{T_1}{T_2}
             & \Gamma \vdash t_1 \tpcheck T_1
             \\ \Gamma \vdash t_2 \tpcheck [t_1/x]^{T_1}\ T_2
             & |t_1| =_{\beta\eta} |t_2|
           \end{array}
          } &
    \infer{\Gamma\vdash t.1 \tpsynth T}{\Gamma\vdash t \tpsynthleads \abs{\iota}{x}{T}{T'}} \\ \\
    \infer{\Gamma\vdash t.2 \tpsynth [t.1/x] T'}{\Gamma\vdash t \tpsynthleads \abs{\iota}{x}{T}{T'}} &
    \infer{\Gamma\vdash \beta\{t'\} \tpcheck T}
          {T \leadstocs \{t_1 \simeq t_2\} \quad \textit{FV}(t')\subseteq \textit{dom}(\Gamma) \quad |t_1| =_{\beta\eta} |t_2|}  \\ \\    
    \infer{\Gamma\vdash \delta\ \mhyph\ t \tpcheck T}
          {
            \Gamma \vdash t \tpsynth T'
            \quad T' \cong \{\absu{\lambda}{x}{\absu{\lambda}{y}{x}} \simeq \absu{\lambda}{x}{\absu{\lambda}{y}{y}}\}
          } &
    \infer{\Gamma \vdash \rho\ t\ @x\langle t_2 \rangle.T'\ \mhyph\ t' \tpcheck T}
          {
          \begin{array}{ccc}
            \Gamma \vdash t \tpsynthleads \{t_1 \simeq t_2'\}
            & \textit{FV}(t_2) \subseteq \textit{dom}(\Gamma)
            & |t_2'| =_{\beta\eta} |t_2|
            \\ \Gamma \vdash [t_2/x] T' \tpsynth \star
            & \Gamma \vdash t' \tpcheck [t_2/x]T'
            & [t_1/x]T' \cong T
          \end{array}
          } \\ \\
    \infer{\Gamma \vdash \chi\ T\ \mhyph\ t\ \tpsynth T}
          {\Gamma \vdash T \tpsynth \star & \Gamma \vdash t \tpcheck T} &
    \infer{\Gamma\vdash \varphi\ t\ \mhyph\ t'\ \{t''\} \Leftrightarrow T}
          {
              \Gamma\vdash t\tpcheck \{t' \simeq t''\}
              & \Gamma\vdash t' \Leftrightarrow T
              & \textit{FV}(t'') \subseteq \textit{dom}(\Gamma)
          }
  \end{array}
    \end{array}
  \]
\caption{Rules for checking a term against a type ($\Gamma \vdash t \tpcheck T$)
           and synthesizing a type for a term ($\Gamma \vdash t \tpsynth T$)}
\label{fig:tp}
\end{figure}

\begin{figure}
  \centering
  \[
    \begin{array}{lll}
      \infer{\vdash \cdot}{}
      & \infer{
        \vdash \Gamma,x:T
        }{\Gamma \vdash T \tpsynth \star}
      & \infer{
        \vdash \Gamma,X:\kappa
        }{
        \Gamma \vdash \kappa
        }
    \end{array}
  \]
  \caption{Rules for checking a context is well-formed}
  \label{fig:ctxt}
\end{figure}

\begin{figure}
  \[
  \begin{array}{lllllll}
    |x| & = & x &\ &
    |\absu{\lambda}{x}{t}| & = & \absu{\lambda}{x}{|t|} \\
    |t\ t'| & = & |t|\ |t'| &\ &
    |t\cdot T| & = & |t| \\
    |\absu{\Lambda}{x}{t}| & = & |t| &\ &
    |t\ \mhyph t'| & = & |t| \\
    |[t , t']| & = & |t| &\ &
    |t.1| & = & |t| \\
    |t.2| & = & |t| &\ &
    |\beta\{t\}| & = & |t|\\
    |\delta\ \mhyph\ t| & = & \absu{\lambda}{x}{x}&\ &
    |\rho\ t\ @x\langle t_2 \rangle.T'\ \mhyph\ t'| & = & |t'| \\
    |\varphi\ t\ \mhyph\ t'\ \{t''\}| & = & |t''| &\ &
    |\chi\ T\ \mhyph\ t| & = & |t|
  \end{array}
  \]
  \caption{Erasure for annotated terms}
  \label{fig:eraser}
\end{figure}  

\begin{figure}
  \[
    \begin{array}{c}
      \fbox{\(T_1 \cong T_2\)}
      \quad \fbox{\(T_1 \cong^{\text{t}} T_2\)}
      \\ \\
      \infer{
       T_1 \cong T_2
      }{
       T_1 \leadstocs T_1' \nleadstoc
       \quad T_2 \leadstocs T_2' \nleadstoc
       \quad T_1' \cong^{\text{t}} T_2'
      }
      \\ \\
      \begin{array}{cc}
        \infer{X \cong^{\text{t}} X}{}
        &
        \infer{
         \abs{\forall}{X}{\kappa_1}{T_1} \cong^{\text{t}} \abs{\forall}{X}{\kappa_2}{T_2}
        }{
         \kappa_1 \cong \kappa_2
         \quad  T_1 \cong T_2
        }
        \\ \\
          \infer{
           \abs{\forall}{x}{T_1}{T_1'} \cong^{\text{t}} \abs{\forall}{x}{T_2}{T_2'}
          }{
           T_1 \cong T_2
           \quad  T_1' \cong T_2'
          }
        & 
        \infer{
         \abs{\Pi}{x}{T_1}{T_1'} \cong^{\text{t}} \abs{\Pi}{x}{T_2}{T_2'}
        }{
         T_1 \cong T_2
         \quad  T_1' \cong T_2'
        }
        \\ \\ 
          \infer{
           \abs{\lambda}{x}{T_1}{T_1'} \cong^{\text{t}} \abs{\lambda}{x}{T_2}{T_2'}
          }{
           T_1 \cong T_2
           \quad  T_1' \cong T_2'
          }
        & 
        \infer{
         \abs{\lambda}{X}{\kappa_1}{T_1} \cong^{\text{t}} \abs{\lambda}{X}{\kappa_2}{T_2}
        }{
         \kappa_1 \cong \kappa_2
         \quad  T_1 \cong T_2
        }
        \\ \\
          \infer{
          \abs{\iota}{x}{T_1}{T_1'} \cong^{\text{t}} \abs{\iota}{x}{T_2}{T_2'}
          }{
          T_1 \cong T_2
          \quad  T_1' \cong T_2'
          }
        & 
        \infer{
         T_1\ t_1 \cong^{\text{t}} T_2\ t_2
        }{
         T_1 \cong^{\text{t}} T_2 \quad |t_1| =_{\beta\eta} |t_2|
        }
        \\ \\
          \infer{
           T_1 \cdot T_1' \cong^{\text{t}} T_2 \cdot T_2'
          }{
           T_1 \cong^{\text{t}} T_2
           \quad T_1' \cong T_2'
          }
        &
        \infer{
         \{ t_1 \simeq t_2 \} \cong^{\text{t}} \{ t_1'\ \simeq t_2' \}
        }{
         |t_1| =_{\beta\eta} |t_1'| \quad |t_2| =_{\beta\eta} |t_2'|
        }
      \end{array}
      \\ \\ \fbox{\(\kappa_1 \cong \kappa_2\)}
      \\ \\
      \begin{array}{cc}
        \infer{
         \star \cong \star
        }{}
        & \infer{
           \abs{\Pi}{x}{T_1}{\kappa_1} \cong \abs{\Pi}{x}{T_2}{\kappa_2}
          }{
           T_1 \cong T_2
           \quad \kappa_1 \cong \kappa_2
          }
        \\ \\
        \infer{
         \abs{\Pi}{X}{\kappa_1}{\kappa_1'} \cong \abs{\Pi}{X}{\kappa_2}{\kappa_2'}
        }{
         \kappa_1 \cong \kappa_2
         \quad \kappa_1' \cong \kappa_2'
        }
      \end{array}
    \end{array}
  \]
  \caption{Conversion rules for classifiers}
  \label{fig:conv}
\end{figure}  

\subsection{Overview of the constructs}
\label{sec:overview}

CDLE has as a subsystem the extrinsic Calculus of
Constructions (CC).  We have dependent function types
$\abs{\Pi}{x}{T}{T'}$ and kinds $\abs{\Pi}{x}{T}{\kappa}$, as well as
term- and type-level quantification over (possibly higher-kinded)
types $\abs{\forall}{X}{\kappa}{T}$ and
$\abs{\Pi}{X}{\kappa}{\kappa'}$.  We use $\forall$ when the
corresponding argument will be erased, and $\Pi$ when it will be
retained.  Since we do not erase term or type arguments from
type-level applications, we thus write $\abs{\Pi}{X}{\kappa}{\kappa'}$
instead of $\abs{\forall}{X}{\kappa}{\kappa'}$.
For abstractions, we write $\lambda$ to
correspond to $\Pi$ and $\Lambda$ to correspond to $\forall$.  As noted
above, application to a type is denoted with center dot ($\cdot$).

To Curry-style CC, CDLE adds: implicit products, introduced originally
by Miquel~\cite{miquel01}; a primitive equality type $\{ t \simeq
t'\}$; and dependent intersection types $\abs{\iota}{x}{T}{T'}$,
introduced by Kopylov~\cite{kopylov03}.  Implicit products are used
for erased arguments to functions, found also in systems like Agda
(c.f.~\cite{mishraLinger08}).  Dependent intersections are a rather
exotic construct allowing us to assign type $\abs{\iota}{x}{T_1}{T_2}$ to
erased term $t$ when we can assign $T_1$ to $t$, and also assign
$[t/x]T$ to $t$.  For an annotated introduction form, we write
$[t_1,t_2$], where $t_1$ checks against type $T_1$, $t_2$ checks against
$[t_1/x]^{T_1}T_2$, and $t_1$ and $t_2$ have $\beta\eta$-equivalent
erasures.  Dependent intersections thus enable a controlled form of
self-reference in the type.  Previous work showed how to use this to
derive induction for Church-encoded natural numbers~\cite{stump18}.

The typing rules include conversion checks in a few places, e.g., as is
standard when switching from checking to synthesizing mode.
For the introduction forms for types, the checked type first is call-by-name
reduced to weak head normal form, which must be formed from the appropriate connective
for the term construct (e.g., a \(\Pi\)-type for a lambda abstraction).
Similarly, for the elimination forms the major premise has its type synthesized
and then call-by-name normalized to reveal the correct connective.
As is standard for a bidirectional type system, we also include the construct
$\chi\ T\ \mhyph\ t$ for type ascription (allowing a term whose type can be
checked to be given a user-provided type so that the whole 
expression synthesizes its type) and a judgmental (or \emph{subsumption}) rule
stating that terms whose types can be synthesized may be checked against a
convertible type.

We have modified the rules for equality types $\{ t_1 \simeq t_2 \}$ so that we require
nothing of $t_1$ and $t_2$ except that the set $\textit{dom}(\Gamma)$ of variables
declared by $\Gamma$ includes their free variables $\textit{FV}(t_1\ t_2)$.  Further modifications
over the version of CDLE in~\cite{stump18} are:
\begin{itemize}
\item To prove $\{ t_1 \simeq t_2\}$ for definitionally equal terms (that is,
  terms that are \(\beta\eta\)-equivalent modulo erasure), one now writes $\beta\{t'\}$,
  with the critical idea that $|\beta\{t'\}|$ erases to (the possibly unrelated) $|t'|$.
  We call this the \textbf{Kleene trick} because it goes back to Kleene's numeric
  realizability~\cite{Kle65_Classical-Extensions-of-Intuitionistic-Mathematics},
  which accepts any number $n$ as a realizer 
  of a true equation.  Here, we accept any term $t'$ as a realizer of $\{ t_1
  \simeq t_2\}$ when \(t_1\) and \(t_2\) are definitionally equal, provided the
  free variables of \(t'\) are declared in the context.
  
  The Kleene trick means that in Cedille, any such term --- even otherwise untypable
  terms, non-normalizing terms, etc. --- prove trivially true equations.
  Put another way, any trivially true equation type in CDLE is a suitable type
  to classify all untyped lambda calculus terms.
\item The \(\rho\) construct allows one to rewrite occurrences of \(t_1\) in the
  checked type \(T\) of the whole expression to \(t_2\) before checking
  the type of the subexpression \(t'\).
  The version presented here requires a type annotation (``guide'') \(@x\langle
  t_2 \rangle.T\).

  In \(\rho\ t\ @x\langle t_2 \rangle.T\ \mhyph\ t'\), the first subexpression
  \(t\) must synthesize (possibly after some normalization) an equation type of
  the form \(\{t_1 \simeq t_2'\}\), and the user provided term \(t_2\) must be
  \(\beta\eta\)-convertible (modulo erasure) to \(t_2'\).
  The user provided type \(T'\) is then checked to have kind \(\star\) after
  replacing occurrences of \(x\) with \(t_2\), and the subexpression \(t'\) is
  checked against this type.
  Since equality is untyped, it may be that \([t_2'/x]T\) is not a well-kinded
  type, so in this explicit form we require a definitionally equal term \(t_2\)
  from the user, which may involve more typing annotations. 
  
  Finally, \([t_1/x]T'\) (which need not be well-kinded) must be convertible
  with the expected type \(T\).
  In Cedille, the guide is optional and the construct may be used to rewrite a
  contextually given type; a heuristic, whose details are beyond the scope of
  this document, is used to produce a resulting type that is well-kinded.
  The current implementation of Cedille (as of April 13, 2021), does not support
  specifying \(t_2\) in the guide \(@x\langle t_2 \rangle.T\).
  This additional specification of the term \(t_2\) was required to prove
  Theorem~\ref{thm:syntactic-kind-pres}, which did not appear in earlier
  versions of this document (and its absence in prior versions does not affect
  the semantic proofs).

\item We adopt a strong form of Nuprl's \textbf{direct computation rules}~\cite{constable+86}:
  If we have a term $t'$ of type $T$ and a proof $t$ that $\{ t' \simeq t''\}$, then we may conclude that
  $t''$ has type $T$ by writing the annotated term $\varphi\ t\ \mhyph\ t'\ \{t''\}$, which
  erases to $t''$.
\item Where the previous version of CDLE uses $\beta$-equivalence for (erased)
  terms, we here adopt $\beta\eta$-equivalence.
  This allows us to observe in many cases that retyping functions are actually
  $\beta\eta$-equivalent to $\absu{\lambda}{x}{x}$. 
  While $\beta\eta$-equivalence takes more work to incorporate into intrinsic
  type theory~\cite{geuvers92}, it raises no difficulties for our extrinsic one.
\item We add an explicit axiom $\delta$ saying that
  Church-encoded boolean \emph{true} is different from \emph{false}.
  In the implementation of Cedille, this is generalized to the rule where the
  proof \(t\) synthesizes an equation \(\{t_1 \simeq t_2\}\) in which \(|t_1|\)
  and \(|t_2|\) are separable using the \emph{B\"ohm-out algorithm}~\cite{BDPR79_Bohm-Algorithm}.

  In the first version of CDLE, such an axiom was derivable from \emph{lifting},
  a construct allowing terms with simple types to be lifted to the type
  level~\cite{stump17}.
  We omit lifting in this new version of CDLE, because while
  sound, lifting as defined in that previous work is complicated and appears to be incomplete.  Developing a new
  form of lifting remains to future work.
\end{itemize}

The equality type remains \textbf{intensional}: we equate closed terms iff they
are $\beta\eta$-equal.

\subsection{Semantics and metatheory}

Figure~\ref{fig:semtp} gives a realizability semantics for types and
kinds, following the semantics given in the previous papers on
CDLE~\cite{stump18,stump17}.  Details of this semantics are presented
further in Section~\ref{sec:snd} below.  Using the semantics and the
definition in Figure~\ref{fig:semctxt} of $\interp{\Gamma}$, we can
prove the following theorem:
\begin{theorem}[Soundness]
\label{thm:snd}
Suppose $(\sigma,\rho)\in\interp{\Gamma}$.  Then we have:
\begin{enumerate}
\item If $\Gamma\vdash \kappa$, then $\interp{\kappa}_{\sigma,\rho}$ is defined.
\item If $\Gamma\vdash T \tpsynth \kappa$, then $\interp{T}_{\sigma,\rho}\in\interp{\kappa}_{\sigma,\rho}$.
\item If $\Gamma\vdash t \tpsynth T$ then $[\sigma |t|]_{\cbe}\in\interp{T}_{\sigma,\rho}\in \mathcal{R}$.
\item If $\Gamma\vdash t \tpcheck T$ and $\interp{T}_{\sigma,\rho}\in \mathcal{R}$, then
    $[\sigma |t|]_{\cbe}\in\interp{T}_{\sigma,\rho}$.
\item If $T \cong T'$ or $T \cong^{\text{t}} T'$ and $\interp{T}_{\sigma,\rho}$ and $\interp{T'}_{\sigma,\rho}$ are both defined, then they are equal.
\item If \(\kappa \cong \kappa'\) and \(\interp{\kappa}_{\sigma,\rho}\) and
  \(\interp{\kappa'}_{\sigma,\rho}\) are both defined, then they are equal.
\end{enumerate}
\end{theorem}

An easy corollary, by the semantics of $\forall$-types, is then:

\begin{theorem}[Logical consistency]
\label{thm:consis}
  There is no term $t$ such that $\cdot \vdash t \tpsynth \abs{\forall}{X}{\star}{X}$.
\end{theorem}

It may worry some readers that we have:
\begin{observation}
  There are typable terms $t$ which fail to normalize.
\end{observation}

Defining \verb|Top| to be $\{\absu{\lambda}{x}{x} \simeq
\absu{\lambda}{x}{x}\}$, we may assign \verb|Top| to any closed term \verb|t|,
including non-normalizing ones. In our annotated syntax, we write
\texttt{\(\beta\)\{t\}}.
Even without this, the presence of \(\varphi\) allows us to type non-normalizing
terms assuming an erased argument \(x\) of type \(\{\absu{\lambda}{x}{x} \simeq
\absu{\lambda}{x}{x\ x}\}\) by changing the type of the term \(\texttt{id}\
\cdot \texttt{True}\ \texttt{id}\), where \texttt{True} is
\(\abs{\forall}{X}{\star}{X \to X}\).
This would allow us to give the type \texttt{True} to \(\Omega = (\absu{\lambda}{x}{x\ x})\
\absu{\lambda}{x}{x\ x}\).
In general, we can use any inconsistent assumption to do this, and in the
presence of \(\delta\) that includes all equations between two terms that are
B\"ohm-separable.
But, failure of normalization does not
impinge on Theorem~\ref{thm:consis}. Extensional Martin-L\"of type theory (MLTT)
is also non-normalizing, for a very similar reason, but this fact does not contradict
its logical soundness~\cite{dybjer16}. In CDLE, the guarantees one gets about
the behavior of terms are expressed almost entirely in their types. If the types
are weak, then not much is guaranteed; but stronger types can guarantee
properties like normalization, as demonstrated by the following theorem:

\begin{theorem}[Call-by-name normalization of functions]
  \label{thm:cedille-termination}
  Suppose \(\cdot \vdash t \tpsynth T\) and \(\cdot \vdash t' \tpsynth T \to
  \abs{\Pi}{x}{T_1}{T_2}\), and furthermore that \(|t'| =
  \absu{\lambda}{x}{x}\).
  Then \(|t|\) is call-by-name normalizing.
\end{theorem}

Given the lack of normalization in general, several checks in the typing rules --
for things like $|t| =_{\beta\eta} |t'|$ -- are formally undecidable.
We simply impose a bound on the number of steps of reduction,
and thus restore formal decidability (we are checking ``typable within
a given budget'').  In practice, the same is done for Coq and Agda,
where type checking is decidable but, in general, infeasible (since one
may write astronomically slow terminating functions).

Finally, in line with ideas recently advocated by Dreyer, we
are less concerned with syntactic
type preservation as we are with \emph{semantic} type
preservation~\cite{dreyer18}.
Note that by construction, semantic types $\interp{T}_{\sigma,\rho}$ are
preserved by $\beta\eta$-reduction:

\begin{theorem}[Semantic type preservation]
  If $t \leadsto_{\beta\eta} t'$ and $t\in\interp{T}_{\sigma,\rho}$, then $t'\in\interp{T}_{\sigma,\rho}$.
\end{theorem}

  Confluence of $\beta\eta$-reduction for (erased)
  terms is nothing other than confluence of untyped lambda calculus.
  This is because, as easily verified by inspecting
  Figure~\ref{fig:eraser}, the erasure function maps annotated terms
  $t$ to terms $|t|$ of pure untyped lambda calculus.
\begin{comment}
  \begin{lemma}
    If $t$ is an annotated term of CDLE, then $|t|$ is a term of pure untyped lambda calculus.
    \end{lemma}
\end{comment}    

We make a concession to syntactic classifier preservation in the case of types
and kinds.
During type inference, types may be reduced using a call-by-name operational
semantics to reveal type constructors.
With the removal of lifting from CDLE, terms cannot compute types and so no
terms need to be reduced during this process.

\begin{theorem}[Syntactic kind preservation]
  \label{thm:syntactic-kind-pres}
  If \(\Gamma \vdash T \tpsynth \kappa\) and \(T \leadstoc T'\) then
  \(\Gamma \vdash T' \tpsynth \kappa'\) for some \(\kappa'\) such that \(\kappa
  \cong \kappa'\).
\end{theorem}

From this theorem and a few other lemmas (see
Appendix~\ref{sec:proof-syntactic-kind-pres}), we can show the \emph{validity}
(or \emph{agreement}) of the judgments comprising CDLE.

\begin{theorem}[Judgment validity]
  \label{thm:judge-valid}
  If \(\vdash \Gamma\) then:
  \begin{enumerate}
  \item if \(\Gamma \vdash T \tpsynth \kappa\) then \(\Gamma \vdash \kappa\)
    
  \item if \(\Gamma \vdash t \tpsynth T\) then \(\Gamma \vdash T \tpsynth \star\)
  \end{enumerate}
\end{theorem}

Type checking \((\Gamma \vdash t \tpcheck T)\) is not covered in
Theorem~\ref{thm:judge-valid}, since the convention for a bidirectional system
is that there \(T\) is already \emph{assumed} to have type \(\star\) under a
typing context.

\begin{figure}
\[
\begin{array}{lll}
\interp{X}_{\sigma,\rho} & = & \rho(X) \\ 
\interp{\Pi x : T_1. T_2}_{\sigma,\rho} & = & 
    [\{ \lambda x.t\ |\ \forall E\in\interp{T_1}_{\sigma,\rho}.\\
\ &\ &\ \ \ \  [[\choice(E)/x]t]_{\cbe}\in\interp{T_2}_{\sigma[x\mapsto \choice(E)],\rho} \ \wedge\ t = |t|\}]_{\cbe}
 \\
\interp{\forall X:\kappa.T}_{\sigma,\rho} & = & 
  \elcap \{ \interp{T}_{\sigma,\rho[X\mapsto S]} |\ S\in\interp{\kappa}_{\sigma,\rho} \}  \\ 
\interp{\forall x:T.T'}_{\sigma,\rho} & = & 
  \elcap_\star \{ \interp{T'}_{\sigma[x\mapsto \choice(E)],\rho}\ |\ E\in\interp{T}_{\sigma,\rho} \} \\ 
\interp{\iota x:T.T'}_{\sigma,\rho} & = & \{ E\in\interp{T}_{\sigma,\rho} |\ E \in \interp{T'}_{\sigma[x\mapsto \choice(E)],\rho} \} \\ 
\interp{\lambda X:\kappa.T}_{\sigma,\rho} & = & (S\in\interp{\kappa}_{\sigma,\rho} \mapsto \interp{T}_{\sigma,\rho[X\mapsto S]}) 
\\ 
\interp{\lambda x:T.T'}_{\sigma,\rho} & = & 
    (E\in\interp{T}_{\sigma,\rho} \mapsto \interp{T'}_{\sigma[x\mapsto \choice(E)],\rho}) 
\\ 
\interp{T\ T'}_{\sigma,\rho} & = & \interp{T}_{\sigma,\rho}(\interp{T'}_{\sigma,\rho})
\\ 
\interp{T\ t}_{\sigma,\rho} & = & \interp{T}_{\sigma,\rho}([\sigma |t|]_{\cbe})
\\
\interp{\{t \simeq t'\}}_{\sigma,\rho} & = & [\{ t''\ |\ \sigma |t| =_{\beta\eta} \sigma |t'|\ \wedge\ t'' = |t''| \}]_{\cbe} \\
\ &\ &\ \ \ \textnormal{ if }\textit{FV}(t\ t')\subseteq\textit{dom}(\sigma) 
\\
\interp{\star}_{\sigma,\rho} & = & \mathcal{R} \\ 
\interp{\Pi x:T.\kappa}_{\sigma,\rho} & = & 
(E\in\interp{T}_{\sigma,\rho} \to \interp{\kappa}_{\sigma[x\mapsto \choice(E)],\rho}),\\
\ &\ &\ \ \ \textnormal{ if }\interp{T}_{\sigma,\rho}\in\mathcal{R} \\
\interp{\Pi x:\kappa.\kappa'}_{\sigma,\rho} & = & (S\in\interp{\kappa}_{\sigma,\rho} \to \interp{\kappa}_{\sigma,\rho[X\mapsto S]}) \\
\elcap_\star X & = & \left\{\begin{array}{l}
                                         \, \negthinspace\elcap X, \textnormal{ if } X\neq\emptyset\\
                                         \, \negthinspace[\mathcal{L}]_{\cbe},\textnormal{ otherwise}
                                       \end{array}\right. \\
\end{array}
\]
\caption{Semantics for types and kinds}
\label{fig:semtp}
\end{figure}

\begin{figure}
\[
\begin{array}{lll}
(\sigma\uplus[x\mapsto t],\rho)\in\interp{\Gamma,x:T} & \Leftrightarrow & (\sigma,\rho)\in\interp{\Gamma} \ \wedge\ 
 [t]_{\cbe}\in\interp{T}_{\sigma,\rho}\in\mathcal{R}\ \wedge\ t = |t| \\
(\sigma,\rho\uplus[X\mapsto S])\in\interp{\Gamma,X:\kappa} & \Leftrightarrow & (\sigma,\rho)\in\interp{\Gamma} \ \wedge\ 
S\in\interp{\kappa}_{\sigma,\rho} \\
(\emptyset,\emptyset)\in\interp{\cdot}
\end{array}
\]
\caption{Semantics of typing contexts $\Gamma$}
\label{fig:semctxt}
\end{figure}

\subsection{Some details about the semantics and the proof of Theorem~\ref{thm:snd}}
\label{sec:snd}

Following the development in~\cite{stump17}, we work with
set-theoretic partial functions for the semantics of higher-kinded
types.  Types are interpreted as $\beta\eta$-closed sets of closed
terms. Let $\mathcal{L}$ be the set of closed terms of pure lambda calculus
(differently from~\cite{stump17}, we include all terms at this point,
even non-normalizing ones).  We
write $=_{\cbe}$ for standard $\beta\eta$-equivalence of pure lambda calculus, restricted to
closed terms; and $[t]_{\cbe}$ for $\{ t'\ |\ t =_{\cbe} t'\}$.  This
is extended to sets $S$ of terms by writing $[S]_{\cbe}$ for
$\{[t]_{\cbe}\ |\ t\in S\}$.  % In a few places we write
% $\textit{nf}(t)$ for the (unique) $\beta\eta$-normal form of term $t$,
% if it has one.
If (in our meta-language) we affirm a statement
involving application of a partial function, then it is to be
understood that that application is defined.

\begin{definition}[Reducibility candidates]
  $\mathcal{R} := \{ [S]_{\cbe}\ |\ S\subseteq \mathcal{L} \}$.
\end{definition}

Throughout the development we find it convenient to use a
\textbf{choice function} $\choice$.  Given any nonempty set $E$ of
terms, $\choice$ returns some element of $E$.  Note that if $a \in A
\in \mathcal{R}$, then $a$ is a nonempty set of terms of pure lambda
calculus; it can also happen that $A \in\mathcal{R}$ is empty.  The
proof of Theorem~\ref{thm:snd} (see appendix) is then a straightforward adaptation of~\cite{stump17}. 

\textbf{Acknowledgments.}  This work was partially supported by the US
NSF support under award 1524519, and US DoD support under award
FA9550-16-1-0082 (MURI program).

%% Bibliography
\bibliography{biblio}

%% Appendix
\appendix

\section{Proof of Theorem~\ref{thm:snd}}

First a few lemmas (easy proofs omitted):

\begin{lemma}
  $\interp{\kappa}_{\sigma,\rho}$ is nonempty if defined.
\end{lemma}

\begin{lemma}
\label{lem:choice}
If $E$ is nonempty, then $[\choice(E)]_{\cbe} = E$
\end{lemma}

\begin{lemma}
  The set $\mathcal{R}$ ordered by subset forms a complete lattice,
  with greatest element $[\mathcal{L}]_{\cbe}$ and greatest lower bound
  of a nonempty set of elements given by
  intersection.  Also, $\emptyset$ is the least element.
\end{lemma}

\begin{lemma}[Term substitution and interpretation]
\label{lem:termsubstinterp}
If $t' =_{\cbe} \sigma |t|$, then: 
\begin{itemize}
\item $\interp{T}_{\sigma[x\mapsto t'],\rho} = \interp{[t/x]T}_{\sigma,\rho}$
\item $\interp{\kappa}_{\sigma[x\mapsto t'],\rho} = \interp{[t/x]\kappa}_{\sigma,\rho}$
\end{itemize}
\end{lemma}

Note that Lemma~\ref{lem:termsubstinterp} also applies to typed substitution: if
\(t' =_{\cbe} \sigma|t|\) then by erasure it is equal to \(\sigma|\chi\ T\
\mhyph\ t'|\).

\begin{lemma}[Type substitution and interpretation] 
  \label{lem:tpsubstinterp}
  \ \\
  \begin{itemize}
  \item $\interp{T}_{\sigma,\rho[X\mapsto\interp{T'}_{\sigma,\rho}]} = \interp{[T'/X]T}_{\sigma,\rho}$
  \item $\interp{\kappa}_{\sigma,\rho[X\mapsto\interp{T'}_{\sigma,\rho}]} = \interp{[T'/X]\kappa}_{\sigma,\rho}$
  \end{itemize}
\end{lemma}

\begin{lemma}
  \label{lem:interppres}
  If $T \leadstocs T'$ and $\interp{T}_{\sigma,\rho}$ is defined, then $\interp{T'}_{\sigma,\rho}$ is also defined and equals $\interp{T}_{\sigma,\rho}$.
\end{lemma}
\begin{proof}
  This follows by induction on the reduction derivation, making use of the previous substitution lemmas.
\end{proof}

\begin{proof}[Soundness (Theorem~\ref{thm:snd})]
  The following proof is adapted from~\cite{stump17}.  It proceeds by
  mutual induction on the assumed typing, kinding, or kind formation
  derivation, for each part of the lemma.  We prove the parts
  successively.  

\subsection{Proof of part (1)}

\startcase{.2cm}
\[
\infer{\Gamma \vdash \star }{\ }
\]
$\interp{\star}_{\sigma,\rho}$ is just $\mathcal{R}$, which is
defined.  

\startcase{.2cm}
\[
    \infer{\Gamma\vdash\abs{\Pi}{x}{T}{\kappa}}{\Gamma \vdash T \tpsynth \star & \Gamma,x:T\vdash\kappa} 
\]
By the IH, $\interp{T}_{\sigma,\rho}\in\mathcal{R}$, and so
$\interp{\Pi x : T.\, \kappa}_{\sigma,\rho}$ is
$(E\in\interp{T}_{\sigma,\rho} \to \interp{\kappa}_{\sigma[x\mapsto  \choice(E)],\rho})$.
The latter quantity is defined if for all
$E\in\interp{T}_{\sigma,\rho}$, $\interp{\kappa}_{\sigma[x\mapsto \choice(E)],\rho})$ is, too.  Since
$\interp{T}_{\sigma,\rho}\in\mathcal{R}$, every element $E$ of
$\interp{T}_{\sigma,\rho}$ is nonempty, as noted above, 
so $\choice(E)$ is defined.  We may apply the IH to the second
premise, since
$(\sigma[x\mapsto\choice(E)],\rho)\in\interp{\Gamma,x:T}$, because $E\in\interp{T}_{\sigma,\rho}$ (by assumption)
and $[\choice(E)]_{\cbe} = E$.  This gives definedness of the semantics
of the $\Pi$-kind.

\startcase{.2cm}
\[
   \infer{\Gamma\vdash\abs{\Pi}{X}{\kappa'}{\kappa}}{\Gamma \vdash \kappa' & \Gamma,X:\kappa'\vdash\kappa}
\]
We must show $(S\in\interp{\kappa}_{\sigma,\rho} \to \interp{\kappa}_{\sigma,\rho[X\mapsto S]})$ is defined.
This is true if $\interp{\kappa}_{\sigma,\rho}$ is defined, which is the case by
the IH applied to the first premise; and if for all
$S\in\interp{\kappa}_{\sigma,\rho}$,
$\interp{\kappa}_{\sigma,\rho[X\mapsto S]}$ is defined.  The latter is
true by the IH applied to the second premise.  

\subsection{Proof of part (2)}

\startcase{.2cm}
\[
\infer{\Gamma \vdash X \tpsynth \kappa}{(X : \kappa) \in \Gamma} 
\]
From the definition of $\interp{\Gamma}$, we obtain
$\rho(X)\in\interp{\kappa}_{\sigma,\rho}$.

\startcase{.2cm}
\[
   \infer{\Gamma\vdash\abs{\Pi}{x}{T}{T'} \tpsynth \star}{\Gamma \vdash T \tpsynth \star & \Gamma,x:T\vdash T' \tpsynth \star} 
\]
We must show $\interp{\Pi x:T.T'}_{\sigma,\rho}\in\mathcal{R}$. The
semantics defines $\interp{\Pi x:T.T'}_{\sigma,\rho}$ to be
$[A]_{\cbe}$ for a certain $A$, where if $A$ is defined, then
$A\subseteq\mathcal{L}$.  So it suffices to shown definedness. By the IH
for the first premise, $\interp{T}_{\sigma,\rho}\in\mathcal{R}$.  This
means that if $E\in\interp{T}_{\sigma,\rho}$, $\choice(E)$ is defined.
We can then apply the IH to the second premise, since
$\sigma[x\mapsto\choice(E)]\in\interp{\Gamma,x:T}$, to obtain
definedness of $\interp{T'}_{\sigma[x\mapsto\choice(E),\rho}$.

\startcase{.2cm}
\[
   \infer{\Gamma\vdash\abs{\forall}{x}{T}{T'} \tpsynth \star}{\Gamma \vdash T \tpsynth \star & \Gamma,x:T\vdash T' \tpsynth \star} 
\]
By the IH for the second premise, $\interp{T_2}_{\sigma[x\mapsto  \choice(E)],\rho}\in\mathcal{R}$, for every
$E\in\interp{T_1}_{\sigma,\rho}$ where
$\interp{T_1}_{\sigma,\rho}\in\mathcal{R}$.  By the IH for the first
premise, we indeed have $\interp{T_1}_{\sigma,\rho}\in\mathcal{R}$.
So if $\interp{T_1}_{\sigma,\rho}$ is non-empty, then the intersection of all the sets
$\interp{T_2}_{\sigma[x\mapsto \choice(E)],\rho}$ where $E\in\interp{T_1}_{\sigma,\rho}$ is a
reducibility candidate, since each of those sets is.  By the semantics
of $\forall$-types quantifying over terms, this is sufficient.  If $\interp{T_1}_{\sigma,\rho}$ is
empty, then the interpretation of the $\forall$-type is $[\mathcal{L}]_{\cbe}$ by the definition
of $\elcap_\star$, and this is in $\mathcal{R}$.

\startcase{.2cm}
\[
    \infer{\Gamma\vdash \abs{\forall}{X}{\kappa}{T} \tpsynth \star}{\Gamma \vdash \kappa & \Gamma,X:\kappa\vdash T \tpsynth \star} 
\]
Similarly to the previous case: by the IH for the second premise,
$\interp{T_2}_{\sigma,\rho[X\mapsto S}\in\mathcal{R}$, for every
$S\in\interp{\kappa}_{\sigma,\rho}$.  By the IH part for the first
premise, $\interp{\kappa}_{\sigma,\rho}$ is defined.  So the
intersection of all the sets $\interp{T_2}_{\sigma,\rho[X\mapsto S]}$
where $S\in\interp{\kappa}_{\sigma,\rho}$ is a reducibility candidate,
since each of those sets is.  The intersection is nonempty, since $\interp{\kappa}_{\sigma,\rho}$ is (as stated in a lemma above).
By the semantics of $\forall$-types
quantifying over types, this is sufficient. 

\startcase{.2cm}
\[
    \infer{\Gamma\vdash\abs{\iota}{x}{T}{T'} \tpsynth \star}{\Gamma \vdash T \tpsynth \star & \Gamma,x:T\vdash T' \tpsynth \star} 
\]
The set $\interp{\iota x:T.T'}_{\sigma,\rho}$ is explicitly defined to
be a subset of $\interp{T}_{\sigma,\rho}$, which is in $\mathcal{R}$,
by the IH applied to the first premise.  Since for any
$A\subseteq\mathcal{L}$, $[A]_{\cbe}$ is in $\mathcal{R}$, to show that
$\interp{\iota x:T.T'}_{\sigma,\rho}$ is also in $\mathcal{R}$ it suffices
to show definedness of $\interp{T'}_{\sigma[x\mapsto \choice(E)],\rho}\}$
(which is used in the predicate picking out the
particular subset of $\interp{T}_{\sigma,\rho}$), for
$E\in\interp{T}_{\sigma,\rho}$.  For such $E$, $\choice(E)$ is defined
(since $\interp{T}_{\sigma,\rho}\in\mathcal{R}$ and hence $E\in\interp{T}_{\sigma,\rho}$ is nonempty)
and in $E$,
so $\sigma[x\mapsto\choice(E)]\in\interp{\Gamma,x:T}$.  So
by the IH for the second premise,
$\interp{T'}_{\sigma[x\mapsto\choice(E),\rho]}$ is defined.

\startcase{.2cm}
\[
 \infer{\Gamma\vdash\abs{\lambda}{x}{T}{T'} \tpsynth \abs{\Pi}{x}{T}{\kappa}}{\Gamma \vdash T \tpsynth \star & \Gamma,x:T\vdash T'\tpsynth\kappa}
\]
By the semantics, $\interp{\lambda x:T.T'}_{\sigma,\rho}$ is
$(E\in\interp{T}_{\sigma,\rho} \mapsto \interp{T'}_{\sigma[x\mapsto
  \choice(E)],\rho})$.  We must show that this (meta-level) function
is in $\interp{\Pi x:T.\kappa}_{\sigma,\rho}$.  By the semantics of
kinds, the latter quantity, if defined, is
$(E\in\interp{T}_{\sigma,\rho} \to_{\cbe}
\interp{\kappa}_{\sigma[x\mapsto \choice(E)],\rho})$.
By the IH for the first premise, $\interp{T}_{\sigma,\rho}\in\mathcal{R}$.
So we must just show that for any $E\in\interp{T}_{\sigma,\rho}$,
$\interp{T'}_{\sigma[x\mapsto
  \choice(E)],\rho}\in\interp{\kappa}_{\sigma[x\mapsto
  \choice(E)],\rho}$.  But this follows by the IH for the second
premise.

\startcase{.2cm}
\[
\infer{\Gamma\vdash\abs{\lambda}{X}{\kappa}{T'} \tpsynth \abs{\Pi}{X}{\kappa}{\kappa'}}
      {\Gamma \vdash \kappa & \Gamma,X:\kappa\vdash T'\tpsynth\kappa'} 
\]
This case is an easier version of the previous one.  It suffices to
assume an arbitrary $S\in\interp{\kappa}_{\sigma,\rho}$ and show
$\interp{T'}_{\sigma,\rho[X\mapsto S]}\in\interp{\kappa'}_{\sigma,\rho[X\mapsto S]}$.  But this follows
by the IH applied to the second premise.  And we have definedness of
$\interp{\kappa}_{\sigma,\rho}$ by the IH for the first premise.

\startcase{.2cm}
\[
   \infer{\Gamma\vdash T\ t \tpsynth [t/x]^{T'}\ \kappa}{\Gamma\vdash T \tpsynth
     \abs{\Pi}{x}{T'}{\kappa} & \Gamma\vdash t \tpcheck T'}  
\]
By the IH for the first premise,
$\interp{T}_{\sigma,\rho}\in\interp{\Pi x:T'.\kappa}_{\sigma,\rho}$.
By the semantics of $\Pi$-kinds, this means that
$\interp{T}_{\sigma,\rho}$ is a function which given any
$E\in\interp{T'}_{\sigma,\rho}$, will produce a result in
$\interp{\kappa}_{\sigma[x\mapsto \choice(E)],\rho}$.  By the
semantics of type applications, $\interp{T\ t}_{\sigma,\rho}$ is equal
to $\interp{T}_{\sigma,\rho}([\sigma |t|]_{\cbe})$.  This is defined,
since $[\sigma |t|]_{\cbe}\in\interp{T'}_{\sigma,\rho}$, by the IH for
the second premise; note that $\interp{T'}_{\sigma,\rho}$ is defined
since otherwise $\interp{\Pi x:T'.\kappa}_{\sigma,\rho}$ would not be defined.
The result of applying the function is thus indeed
in $\interp{[t/x]^{T'}\ \kappa}_{\sigma,\rho}$, since \(|\chi\ T'\ \mhyph\ t| =
|t|\) (recall the shorthand \([t/x]^{T'} = [\chi\ T'\ \mhyph\ t/x]\)), and with  
Lemma~\ref{lem:termsubstinterp} we have that the interpretation equals
$\interp{\kappa}_{\sigma[x\mapsto \choice([\sigma |t|]_{\cbe})],\rho}$
(the codomain of the function being applied).

\startcase{.2cm}
\[
  \infer{\Gamma\vdash T_1 \cdot T_2 \tpsynth [T_2/X]\kappa_1}{\Gamma\vdash T_1 \tpsynth \abs{\Pi}{X}{\kappa_2}{\kappa_1} & \Gamma\vdash T_2 \tpsynth \kappa_2' & \kappa_2 \cong \kappa_2'}
\]
By the IH applied to the first premise,
$\interp{T_1}_{\sigma,\rho}\in\interp{\abs{\Pi}{X}{\kappa_2}{\kappa_1}}_{\sigma,\rho}$.
By the semantics of $\Pi$-kinds, this means that for any
$S\in\interp{\kappa_2}_{\sigma,\rho}$, $\interp{T_1}_{\sigma,\rho}\ S$ is in
$\interp{\kappa_1}_{\sigma,\rho[X\mapsto S]}$.  By the IH for the second premise, we have 
$\interp{T_2} \in \interp{\kappa_2'}_{\sigma,\rho}$, and by the IH for the third premise,
we have $\interp{\kappa_2}_{\sigma,\rho} = \interp{\kappa_2'}_{\sigma,\rho}$.  So
we get $\interp{T_1}_{\sigma,\rho}(\interp{T_2}_{\sigma,\rho})\in \interp{\kappa_1}_{\sigma,\rho[X\mapsto \interp{T'}_{\sigma,\rho}]}$,
which suffices by Lemma~\ref{lem:tpsubstinterp}.

\startcase{.2cm}
\[
    \infer{\Gamma\vdash \{ t \simeq t' \} : \star}{\textit{FV}(t\ t')\subseteq\textit{dom}(\Gamma)}
\]
Either $\sigma |t| =_{\cbe} \sigma |t'|$ or not.  Either way, the interpretation is defined and in $\mathcal{R}$, since
$\textit{FV}(t\ t')\subseteq\textit{dom}(\sigma)$ (as an easy consequence of $(\sigma,\rho)\in\interp{\Gamma}$).

\subsection{Proof of parts (3) and (4)}

\startcase{.2cm}
\[
    \infer{\Gamma\vdash x\tpsynth T}{(x : T)\in\Gamma} 
\]
This follows from the definition of $\interp{\Gamma}$.

% \startcase{.2cm}
% \[
%  \infer{\Gamma\vdash t\tpcheck T'}{\Gamma\vdash t \tpcheck T & T' \leadsto^*_\beta T} 
% \]
% We are assuming $\interp{T'}_{\sigma,\rho}$ is defined, since this is a checking judgment. The desired result then
% follows from Lemma~\ref{lem:interppres}.

% \startcase{.2cm}
% \[
%     \infer{\Gamma\vdash t\tpsynth T'}{\Gamma\vdash t \tpsynth T & T \leadsto^*_\beta T'} 
% \]
% This also follows from Lemma~\ref{lem:interppres} and the induction hypothesis for the first premise,
% which implies $\interp{T}_{\sigma,\rho}\in\mathcal{R}$ (and hence defined).

\startcase{.2cm}
\[
  \infer{\Gamma\vdash t\tpcheck T}{\Gamma\vdash t\tpsynth T' & T' \cong T}
\]
By the IH applied to the first premise, we have 
$[\sigma|t|]_{\cbe}\in\interp{T'}_{\sigma,\rho}\in\mathcal{R}$. By assumption, $\interp{T}_{\sigma,\rho}\in\mathcal{R}$,
and so by the IH applied to the second premise, we have $[\sigma|t|]_{\cbe}\in\interp{T'}_{\sigma,\rho} = \interp{T}_{\sigma,\rho}$.

\startcase{.2cm}
\[
  \infer{\Gamma\vdash \absu{\lambda}{x}{t} \tpcheck T}{T \leadstocs \abs{\Pi}{x}{T_1}{T_2} & \Gamma,x:T_1\vdash t\tpcheck T_2}
\]
To show $[\sigma \lambda
x.|t|]_{\cbe}\in\interp{\abs{\Pi}{x}{T_1}{T_2}}_{\sigma,\rho}$ (noting that the
latter is defined and in $\mathcal{R}$ by assumption), it suffices to assume an
arbitrary $E\in\interp{T_1}_{\sigma,\rho}$, and show
$[[\choice(E)/x]\sigma|t|]_{\cbe}\in\interp{T_2}_{\sigma[x\mapsto\choice(E)],\rho}$.
By the IH, we have
$[\sigma[x\mapsto\choice(E)]|t|]_{\cbe}\in\interp{T_2}_{\sigma[x\mapsto\choice(E)],\rho}$.
But $[\sigma[x\mapsto\choice(E)]t]_{\cbe} = [[\choice(E)/x]\sigma|t|]_{\cbe}$,
so this is sufficient.

\startcase{.2cm}
\[
  \infer{\Gamma\vdash t\ t' \tpsynth [t'/x]^{T'}T}{\Gamma\vdash t \tpsynthleads \abs{\Pi}{x}{T'}{T} & \Gamma\vdash t' \tpcheck T'} \\ \\
\]
By the IH applied to the first premise, $[\sigma
|t|]_{\cbe}\in\interp{\Pi x:T'.T}_{\sigma,\rho}\in\mathcal{R}$.  This
means that there exists a $\lambda$-abstraction $\lambda x.\hat{t}$
such that $\lambda x.\hat{t} =_{\cbe} \sigma |t|$, by the semantics of $\Pi$-types.
Furthermore, for any $E\in\interp{T'}_{\sigma,\rho}$,
$[[\choice(E)/x]\hat{t}]_{\cbe}\in\interp{T}_{\sigma[x\mapsto\choice(E)],\rho}$.
By the IH applied to the second premise, $[\sigma |t'|]_{\cbe}\in\interp{T'}_{\sigma,\rho}$,
so we can instantiate the quantifier in the previous formula to obtain
\[
 [[\choice([\sigma |t'|]_{\cbe})/x]\hat{t}]_{\cbe}\in\interp{T}_{\sigma[x\mapsto\choice([\sigma |t'|]_{\cbe})],\rho}
\]
By Lemma~\ref{lem:termsubstinterp}, this is equivalent to
\[
 [[\choice([\sigma |t'|]_{\cbe})/x]\hat{t}]_{\cbe}\in\interp{[t'/x]T_2}_{\sigma,\rho}
\]
Since $\sigma |t\ t'| =_{\cbe} (\lambda x.\hat{t})\ \sigma |t'| =_{\cbe} [[\choice([\sigma |t'|]_{\cbe})/x]\hat{t}$,
this is sufficient.

\startcase{.2cm}
\[
  \infer{\Gamma\vdash \absu{\Lambda}{X}{t} \tpcheck T'}
  {T' \leadstocs \abs{\forall}{X}{\kappa}{T} & \Gamma,X:\kappa\vdash t \tpcheck T}
\]
By the IH, $[\sigma |t|]_{\cbe}\in\interp{T}_{\sigma,\rho[X\mapsto S]}$, for all $S\in\interp{\kappa}_{\sigma,\rho}$.
This is sufficient to prove $[\sigma |\absu{\Lambda}{X}{t}|]_{\cbe}\in\interp{\forall X:\kappa.T}_{\sigma,\rho}$, by the semantics
of $\forall$-types and definition of erasure.

\startcase{.2cm}
\[
  \infer{\Gamma\vdash t \cdot T' \tpsynth [T'/X]T}
  {\Gamma\vdash t \tpsynthleads \abs{\forall}{X}{\kappa}{T} & \Gamma\vdash T' \tpsynth \kappa' & \kappa'\cong\kappa} \\ \\
\]
By the semantics of $\forall$-types and the IH applied to the first
premise, we have $[\sigma |t|]_{\cbe}\in\interp{T}_{\sigma,\rho[X\mapsto
  S]}$, for all $S\in\interp{\kappa}_{\sigma,\rho}$.
By applying the IH twice, once to the second premise and once to the third, we
have, we have \(\interp{T'}_{\sigma,\rho} \in \interp{\kappa}_{\sigma,\rho}\).
So, can derive $[\sigma |t|]_{\cbe}\in\interp{T}_{\sigma,\rho[X\mapsto \interp{T'}_{\sigma,\rho}]}$.
By Lemma~\ref{lem:tpsubstinterp},
this is equivalent to the required $[\sigma |t|]_{\cbe}\in\interp{[T'/X]T}_{\sigma,\rho}$,
using also the definition of erasure.

\startcase{.2cm}
\[
  \infer{\Gamma\vdash \absu{\Lambda}{x}{t} \tpcheck T}
  {T \leadstocs \abs{\forall}{x}{T_1}{T_2} & \Gamma,x:T_1\vdash t \tpcheck T_2 & x\not\in\textit{FV}(|t|)}
\]

By the IH applied to the second premise, we have
$[\sigma[x\mapsto\choice(E)]
|t|]_{\cbe}\in\interp{T_2}_{\sigma[x\mapsto\choice(E)],\rho}$, for any
$E\in\interp{T_1}_{\sigma,\rho}$.
This is because $\interp{T_1}_{\sigma,\rho}\in\mathcal{R}$, since
$\interp{\abs{\forall}{x}{T_1}{T_2}}_{\sigma,\rho}$ is in $\mathcal{R}$ and 
hence defined, by assumption.
Since $x\not\in\textit{FV}(|t|)$, we know
$[[\sigma[x\mapsto\choice(E)]|t|]_{\cbe} = [\sigma |t|]_{\cbe}$.
By the semantics of $\forall$-types and definition of erasure, this suffices to
show the desired conclusion.

\startcase{.2cm}
\[
  \infer{\Gamma\vdash t\ \mhyph t' \tpsynth [t'/x]^{T'}T}{\Gamma\vdash t \tpsynthleads \abs{\forall}{x}{T'}{T} & \Gamma\vdash t' \tpcheck T'} \\ \\
\]
The result follows easily by the IH applied to the premises, the
semantics of $\forall$-types, definition of erasure, and Lemma~\ref{lem:termsubstinterp}.

\startcase{.2cm}
\[
  \infer{\Gamma\vdash [ t_1 , t_2 ] \tpcheck T}
  {
    \begin{array}{cc}
      T \leadstocs \abs{\iota}{x}{T_1}{T_1}
      & \Gamma \vdash t_1 \tpcheck T_1
      \\ \Gamma \vdash t_2 \tpcheck [t/x]^{T_1}\ T_2
      & |t_1| =_{\beta\eta} |t_2|
    \end{array}
  }
\]
By the IH, we have $[\sigma |t_1|]_{\cbe}\in\interp{T_1}_{\sigma,\rho}$ and
$[\sigma |t_2|]_{\cbe}\in\interp{[t_1/x]T_2}_{\sigma,\rho}$.
By Lemma~\ref{lem:termsubstinterp}, the latter is equivalent to
$[\sigma |t|]_{\cbe}\in\interp{T_2}_{\sigma[x\mapsto\choice([\sigma
  |t|]_{\cbe}),\rho}$.
These two facts about $[\sigma |t|]_{\cbe}$ are
sufficient, by the semantics of $\iota$-types, for the desired
conclusion, using also the fact (from the fourth premise) that $\sigma|t| =_{\cbe} \sigma|t'|$.

\startcase{.2cm}
\[
  \infer{\Gamma\vdash t.1 \tpsynth T}{\Gamma\vdash t \tpsynthleads \abs{\iota}{x}{T}{T'}}
\]
The desired conclusion follows easily from the IH and the semantics of $\iota$-types.

\startcase{.2cm}
\[
  \infer{\Gamma\vdash t.2 \tpsynth [t.1/x] T'}{\Gamma\vdash t \tpsynthleads \abs{\iota}{x}{T}{T'}}
\]
Similar to the previous case, additionally using Lemma~\ref{lem:termsubstinterp}.

\startcase{.2cm}
\[
  \infer{\Gamma\vdash \beta\{t'\} \tpcheck T}
  {T \leadstocs \{t_1 \simeq t_2\} \quad \textit{FV}(t')\subseteq \textit{dom}(\Gamma) \quad |t_1| =_{\beta\eta} |t_2|}
\]
By the third premise \(|t_1| =_{\beta\eta} |t_2|\) it follows that \(\sigma|t_1|
=_{\beta\eta} \sigma|t_2|\), and also by the second premise and the fact that
\(\sigma \in \interp{\Gamma}\) it follows that
\(\textit{FV}(t_1\ t_2) \subseteq \textit{dom}(\sigma)\).
Also, we see \(\sigma|t'|\) is a closed term.
So, $[\sigma|t'|]_{\cbe}\in\interp{\{ t_1 \simeq t_2 \}}_{\sigma,\rho}$ follows
directly from the semantics of equality types.

\startcase{.2cm}
\[
  \infer{\Gamma\vdash \delta\ \mhyph\ t \tpcheck T}
  {
    \Gamma \vdash t \tpsynth T'
    \quad T' \cong \{\absu{\lambda}{x}{\absu{\lambda}{y}{x}} \simeq \absu{\lambda}{x}{\absu{\lambda}{y}{y}}\}
  }
    %infer{\Gamma\vdash \delta\ \mhyph\ t \tpcheck T}{\Gamma\vdash t\tpsynth \{ \absu{\lambda}{x}{\absu{\lambda}{y}{x}} \simeq \absu{\lambda}{x}{\absu{\lambda}{y}{y}}\}}  
\]
By the inductive hypothesis, \([\sigma|t|]_{\cbe} \in \interp{T'} \in
\mathcal{R}\).
It is easy to see that the meaning of \(\{\absu{\lambda}{x}{\absu{\lambda}{y}{x}
\simeq \absu{\lambda}{x}{\absu{\lambda}{y}}}\}\) exists and is in
\(\mathcal{R}\), and is in fact the empty set.
By mutual induction on the second premise we have \([\sigma|t|]_{\cbe} \in
\interp{\{\absu{\lambda}{x}{\absu{\lambda}{y}{x} \simeq
    \absu{\lambda}{x}{\absu{\lambda}{y}}}\}}_{\sigma,\rho}\), a contradiction.

\startcase{.2cm}
\[
  \infer{\Gamma \vdash \rho\ t\ @x\langle t_2 \rangle.T'\ \mhyph\ t' \tpcheck T}
  {
    \begin{array}{ccc}
      \Gamma \vdash t \tpsynthleads \{t_1 \simeq t_2'\}
      & \textit{FV}(t_2) \subseteq \textit{dom}(\Gamma)
      & |t_2'| =_{\beta\eta} |t_2|
      \\ \Gamma \vdash [t_2/x] T' \tpsynth \star
      & \Gamma \vdash t' \tpcheck [t_2/x]T'
      & [t_1/x]T' \cong T
    \end{array}
  }
\]
By the IH applied to the first premise in the first row, $\sigma|t_1| =_{\cbe}
\sigma|t_2'|$.
With the second and third premise, and from the assumption \(\sigma \in
\interp{\Gamma}\), we have \(\sigma|t_2'| =_{\cbe} \sigma|t_2|\) (and
\(\sigma|t_2|\) is closed), so we obtain
\([\sigma|t_1|]_{\cbe} = [\sigma|t_2|]_{\cbe}\).
By the IH applied to the first premise of the second row,
\(\interp{[t_2/x]T}_{\sigma,\rho} \in \mathcal{R}\) and so is defined.
By the IH applied to the second premise in the second row,
\([\sigma|t'|]_{c\beta\eta} \in \interp{[t_2/x]T'}_{\sigma,\rho}\).
By the IH applied to the second premise in the second row,
\(\interp{T}_{\sigma,\rho} = \interp{[t_1/x]T'}_{\sigma,\rho}\).
The result then follows by applying Lemma~\ref{lem:termsubstinterp}.

\startcase{.2cm}
\[
  \infer{\Gamma \vdash \chi\ T\ \mhyph\ t\ \tpsynth T}
  {\Gamma \vdash T \tpsynth \star & \Gamma \vdash t \tpcheck T}
\]
We apply the IH for the first premise to get that
$\interp{T}_{\sigma,\rho}$ is in $\mathcal{R}$ and hence defined.
Using this, we apply the IH on the second premise to get \([\sigma|t|]_{\cbe}
\in \interp{T}_{\sigma,\rho}\), which by the definition of erasure is what we
must show.

\startcase{.2cm}
\[
  \infer{\Gamma\vdash \varphi\ t\ \mhyph\ t'\ \{t''\} \Leftrightarrow T}
  {
    \Gamma\vdash t\tpcheck \{t' \simeq t''\}
    & \Gamma\vdash t' \Leftrightarrow T
    & \textit{FV}(t'') \subseteq \textit{dom}(\Gamma)
  }
\]
By the IH for the first premise, $\sigma|t'| =_{\cbe} \sigma|t''|$ (we can see
that \(\textit{FV}(t'\ t'') \subseteq \Gamma\), so \(\interp{\{t' \simeq
  t''\}}_{\sigma,\rho}\) is defined).
By the IH for the second premise,
$[\sigma|t'|]_{\cbe}\in\interp{T}_{\sigma,\rho}$ (in the case of type synthesis,
the IH also tells us \(\interp{T}_{\sigma,\rho}\) is defined).
This suffices for the desired conclusion, using also the definition of erasure
($|\varphi\ t\ \mhyph\ t'\{t''\}| = |t''|$).

\subsection*{Proof of part (5) }
We show the cases for the non-congruential rules of Figure~\ref{fig:conv}.

\startcase{.2cm}
\[
  \infer{
    T_1 \cong T_2
  }{
    T_1 \leadstocs T_1' \nleadstoc
    \quad T_2 \leadstocs T_2' \nleadstoc
    \quad T_1' \cong^{\text{t}} T_2'
  }
\]
By Lemma~\ref{lem:interppres}, we have
\[
\begin{array}{lll}
  \interp{T_1}_{\sigma,\rho} & = & \interp{T_1'}_{\sigma,\rho}\\
  \interp{T_2}_{\sigma,\rho} & = & \interp{T_2'}_{\sigma,\rho}
\end{array}
\]
By the IH for the third premise, we have $\interp{T_1'}_{\sigma,\rho} = \interp{T_2'}_{\sigma,\rho}$,
which suffices.

\startcase{.2cm}
\[
    \infer{T_1\ t_1 \cong^{\text{t}} T_2\ t_2}{T_1 \cong^{\text{t}} T_1 & |t_1| =_{\beta\eta} |t_2|}
\]
By the semantics, $\interp{T_1\ t_1}_{\sigma,\rho} =
\interp{T_1}_{\sigma,\rho}([\sigma|t_1|]_{\cbe})$.  By the second premise
and the IH for the first premise, this equals
$\interp{T_2}_{\sigma,\rho}([\sigma|t_2|]_{\cbe})$, as required.

\startcase{.2cm}
\[
    \infer{\{ t_1 \simeq t_2 \} \cong^{\text{t}} \{ t_1'\ \simeq t_2' \}}{|t_1| =_{\beta\eta} |t_1'| & |t_2| =_{\beta\eta} |t_2'|}
\]
This follows easily from the premises and the semantics of equality types.

\subsection*{Proof of Part (6)}

The convertibility relation for kinds consists entirely of congruential rules
for quantification over types and kinds.

\end{proof}

\section{Proof of Theorem~\ref{thm:cedille-termination}}
\begin{proof}[Proof (of Theorem~\ref{thm:cedille-termination})]
Theorem \ref{thm:snd} implies that since
$t'\ t$ is closed and of type $\Pi x:T_1.\ T_2$, we have
$[|t'\ t|]_{c\beta\eta}\in \interp{\Pi x:T_1.\ T_2}_{\cdot,\rho}$,
where $[|t'\ t|]_{c\beta\eta}$ is the set of closed terms which are
$\beta\eta$-equivalent to $|t'\ t|$; and
$(\cdot,\rho)\in\interp{\Gamma}$ gives interpretations $\cdot$ for
term- and $\rho$ for type-variables in $\Gamma$.  By the semantics of
types defined in Figure \ref{fig:semtp}, the
interpretation of a $\Pi$-type consists of sets of the form
$[\lambda x.\,t']_{c\beta\eta}$.  So we have that $|t'\ t|$ is $\beta\eta$-equivalent to
$\lambda x.\, t'$ for some $x, t'$.  
Since $|t'\ t|$ is $\beta$-equivalent to $t$, we know $t =_{\beta\eta} \lambda x.\, t'$.
It is then an easy consequence of the
standardization theorem for untyped lambda calculus that $|t|$ is
call-by-name normalizing (cf.~\cite{Kashima2000}).
\end{proof}

\section{Proof of Theorems~\ref{thm:syntactic-kind-pres} and \ref{thm:judge-valid}}
\label{sec:proof-syntactic-kind-pres}

First a few lemmas (easy cases omitted):

\begin{lemma}
  \label{lem:ctxt-conv-class}
  Let \(\kappa_1,\kappa_1'\) be kinds such that \(\kappa_1 \cong \kappa_1'\).
  \begin{itemize}
  \item If \(\Gamma_1,X:\kappa_1,\Gamma_2 \vdash \kappa_2\) and \(\Gamma_1
    \vdash \kappa_1'\) then \(\Gamma_1,X:\kappa_1',\Gamma_2 \vdash \kappa_2\)
   
  \item If \(\Gamma_1,X:\kappa_1,\Gamma_2 \vdash T \tpsynth \kappa_2\) and
    \(\Gamma_1 \vdash \kappa_1'\) then \(\Gamma_1,X:\kappa_1',\Gamma_2 \vdash T
    \tpsynth \kappa'_2\) for some \(\kappa_2' \cong \kappa_2\)

  \item If \(T_3 \cong T_4\) and \(\Gamma_1,X:\kappa_1,\Gamma_2 \vdash t
    \tpcheck T_3\) then \(\Gamma_1,X:\kappa_1',\Gamma_2 \vdash t \tpcheck T_4\)
    
  \item If \(\Gamma_1,X:\kappa_1,\Gamma_2 \vdash t \tpsynth T\) and
    \(\Gamma_1 \vdash \kappa_1'\) then \(\Gamma_1,X:\kappa_1',\Gamma_2 \vdash t
    \tpsynth T\).
  \end{itemize}
  Furthermore, the resulting typing derivations have depths no
  larger than the assumed ones.
\end{lemma}

\begin{proof}
  By mutual induction on the assumed derivation, and mutually with
  Lemma~\ref{lem:ctxt-conv-class2}.
  The measure, which omits the size of derivations of the convertibility
  relation, is used to ensure that each mutually inductive call is well-founded.
  We show the interesting cases.

  % kinding
  \startcase{.2cm}
  \[
    \infer{\Gamma_1,X_1:\kappa_1,\Gamma_2 \vdash X \tpsynth \kappa}{(X : \kappa) \in \Gamma_1,X_1:\kappa_1,\Gamma_2}
  \]
  We have two subcases to consider.
  If \(X = X_1\), then by the rule we obtain \(\Gamma_1,X_1:\kappa_1',\Gamma_2
  \vdash X_1 \tpsynth \kappa_1'\), as desired.
  Otherwise, by the rule we obtain \(\Gamma_1,X_1:\kappa_1',\Gamma_2
  \vdash X \tpsynth \kappa\).
  In both cases, the resulting derivation has the same depth (i.e., 1) as we
  started with.

  \startcase{.2cm}
  \[
    \infer{
      \Gamma_1,X_1:\kappa_1,\Gamma_2 \vdash\abs{\lambda}{x}{T}{T'} \tpsynth
      \abs{\Pi}{x}{T}{\kappa}
    }{
      \Gamma_1,X_1:\kappa_1,\Gamma_2 \vdash T \tpsynth \star
      & \Gamma_1,X_1:\kappa_1,\Gamma_2 ,x:T\vdash T'\tpsynth\kappa
    }
  \]
  By the IH on the first premise, \(\Gamma_1,X_1:\kappa_1',\Gamma_2 \vdash T
  \tpsynth \star\) at no greater depth (note \(\star\) is convertible only with
  itself).
  By the IH on the second premise, \(\Gamma_1,X_1:\kappa_1',\Gamma_2,x:T \vdash
  T' \tpsynth \kappa'\) (at no greater depth) for some \(\kappa' \cong \kappa\).
  By the rule, we have \(\Gamma_1,X_1:\kappa_1',\Gamma_2 \vdash
  \abs{\lambda}{x}{T}{T'} \tpsynth \abs{\Pi}{x}{T}{\kappa'}\), which is
  congruent with the original type in the conclusion.
  We see that the measure is preserved.

  \startcase{.2cm}
  \[
    \infer{
      \Gamma_1,X_1:\kappa_1,\Gamma_2 \vdash \abs{\lambda}{X}{\kappa}{T'} \tpsynth
      \abs{\Pi}{X}{\kappa}{\kappa'}
    }{
      \Gamma_1,X_1:\kappa_1,\Gamma_2 \vdash \kappa
      & \Gamma_1,X_1:\kappa_1,\Gamma_2,X:\kappa \vdash  T' \tpsynth \kappa'
    }
  \]
  By the IH on the first premise, \(\Gamma_1,X_1:\kappa_1',\Gamma_2 \vdash
  \kappa\).
  By the IH on the second premise, \(\Gamma_1,X_1:\kappa_1',\Gamma_2,X:\kappa
  \vdash T' \tpsynth \kappa''\) for some \(\kappa'' \cong \kappa'\).
  By the rule, \(\Gamma_1,X_1:\kappa_1',\Gamma_2 \vdash
  \abs{\lambda}{X}{\kappa}{T'} \tpsynth \abs{\Pi}{X}{\kappa}{\kappa''}\), which
  is convertible with the original kind synthesized for this type.
  Since the depths of the premises for the resulting derivation are the same as
  the depths for the premises of the assumed one, the measure is preserved.

  \startcase{.2cm}
  \[
    \infer{
      \Gamma_1,X_1:\kappa_1,\Gamma_2 \vdash T\ t \tpsynth [t/x]^{T'}\kappa
    }{
      \Gamma_1,X_1:\kappa_1,\Gamma_2 \vdash T \tpsynth \abs{\Pi}{x}{T'}{\kappa}
      & \Gamma_1,X_1:\kappa_1,\Gamma_2 \vdash t \tpcheck T'
    }
  \]
  By the IH on the first premise, \(\Gamma_1,X_1:\kappa_1',\Gamma_2 \vdash T
  \tpsynth \abs{\Pi}{x}{T''}{\kappa''}\) for some \(T'' \cong T'\) and
  \(\kappa'' \cong \kappa\).
  By the IH on the second premise, \(\Gamma_1,X_1:\kappa_1',\Gamma_2 \vdash t
  \tpcheck T''\).
  By the rule, \(\Gamma_1,X_1:\kappa_1',\Gamma_2 \vdash T\ t \tpsynth
  [t/x]^{T''}\kappa''\), convertible with the given type \([t/x]^{T'}\kappa\).
  We see that the measure is preserved.

  \startcase{.2cm}
  \[
    \infer{
      \Gamma_1,X_1:\kappa_1,\Gamma_2 \vdash T \cdot T' \tpsynth
      [T'/X]\kappa
    }{
      \Gamma_1,X_1:\kappa_1,\Gamma_2 \vdash T \tpsynth
      \abs{\Pi}{X}{\kappa'}{\kappa}
      & \Gamma_1,X_1:\kappa_1,\Gamma_2 \vdash T' \tpsynth \kappa''
      & \kappa' \cong \kappa''}
  \]
  By the IH on the first premise, \(\Gamma_1,X_1:\kappa_1,\Gamma_2 \vdash T
  \tpsynth \abs{\Pi}{X}{\kappa_2'}{\kappa_2}\) for some \(\kappa_2' \cong
  \kappa'\), \(\kappa_2 \cong \kappa\).
  By the IH on the second premise, \(\Gamma_1,X_1:\kappa_1,\Gamma_2 \vdash T'
  \tpsynth \kappa_2''\) for some \(\kappa_2'' \cong \kappa''\).
  By transitivity and the third premise, \(\kappa_2' \cong \kappa_2''\).
  By the rule, \(\Gamma_1,X_1:\kappa_1',\Gamma_2 \vdash T \cdot T' \tpsynth
  [T'/X]\kappa_2\), with this kind convertible to \([T'/X]\kappa\) by
  congruence.
  We see the measure is preserved.

  % checking
  \startcase{.2cm}
  \[
    \infer{
      \Gamma_1,X_1:\kappa_1,\Gamma_2 \vdash t\tpcheck T_3
    }{
      \Gamma_1,X_1:\kappa_1,\Gamma_2 \vdash t\tpsynth T'
      & T' \cong T_3
    }
  \]
  By the IH on the first premise, \(\Gamma_1,X_1:\kappa_1',\Gamma_2 \vdash t
  \tpsynth T'\).
  By assumption and transitivity of congruence, \(T' \cong T_4\) (where \(T_4\)
  is the type given to us in the proof).
  By the rule, \(\Gamma_1,X_1:\kappa_1',\Gamma_2 \vdash t \tpcheck T_4\), and we
  see the measure is preserved.

  \startcase{.2cm}
  \[
    \infer{
      \Gamma_1,X_1:\kappa_1,\Gamma_2 \vdash \absu{\lambda}{x}{t} \tpcheck T_3
    }{
      T_3 \leadstocs \abs{\Pi}{x}{T_1}{T_2}
      & \Gamma_1,X_1:\kappa_1,\Gamma_2,x:T_1\vdash t\tpcheck T_2
    }
  \]
  By the first premise, \(T_3 \cong \abs{\Pi}{x}{T_1}{T_2}\).
  By assumption and transitivity of the convertibility relation, \(T_4 \cong
  \abs{\Pi}{x}{T_1}{T_2}\).
  This means \(T_4 \leadstocs \abs{\Pi}{x}{T_1'}{T_2'}\) for some \(T_1'
  \cong T_1\) and \(T_2' \cong T_2\).
  By mutual induction with Lemma~\ref{lem:ctxt-conv-class2} on the second
  premise, we have a derivation of \(\Gamma_1,X_1:\kappa_1,\Gamma_2,x:T_1'
  \vdash t \tpcheck T_2'\) that is no deeper than that premise.
  So, we are entitled to use the IH on this new derivation to obtain
  \(\Gamma_1,X_1:\kappa_1',\Gamma_2,x:T_1' \vdash t \tpcheck T_2'\), which is
  also no deeper.
  By the rule, \(\Gamma_1,X_1:\kappa_1',\Gamma_2 \vdash \absu{\lambda}{x}{t}
  \tpcheck T_4\), and the measure is preserved.
  
  \startcase{.2cm}
  \[
    \infer{
      \Gamma_1,X_1:\kappa_1,\Gamma_2\vdash \absu{\Lambda}{X}{t} \tpcheck T_3
    }{
      T_3 \leadstocs \abs{\forall}{X}{\kappa}{T}
      & \Gamma_1,X_1:\kappa_1,\Gamma_2,X:\kappa\vdash t \tpcheck T
    }
  \]
  By the first premise, \(T_3 \cong \abs{\forall}{X}{\kappa}{T}\), so by
  assumption and transitivity of convertibility, \(T_4 \cong
  \abs{\forall}{X}{\kappa}{T}\).
  This means \(T_4 \leadstocs \abs{\forall}{X}{\kappa'}{T'}\) for some
  \(\kappa' \cong \kappa\) and \(T' \cong T\).
  We apply the IH once on the second premise to obtain a derivation of
  \(\Gamma_1,X_1:\kappa_1,\Gamma_2,X:\kappa' \vdash t \tpcheck T'\), and then
  again to obtain \(\Gamma_1,X_1:\kappa_1',\Gamma_2,X:\kappa' \vdash t \tpcheck
  T'\), noting this derivation's depth is no greater than the second premise.
  By the rule, \(\Gamma_1,X_1:\kappa_1,\Gamma_2 \vdash \absu{\Lambda}{X}{t}
  \tpcheck T_4\), and the measure is preserved.

  \startcase{.2cm}
  \[
    \infer{
      \Gamma_1,X_1:\kappa_1,\Gamma_2 \vdash \beta\{t'\} \tpcheck T_3
    }{
      T_3 \leadstocs \{t_1 \simeq t_2\}
      & \textit{FV}(t')\subseteq \textit{dom}(\Gamma_1,X_1:\kappa_1,\Gamma_2)
      & |t_1| =_{\beta\eta} |t_2|
    }
  \]
  From the first premise, assumption, and transitivity of convertibility, \(T_4
  \cong \{t_1 \simeq t_2\}\).
  This means that \(T_4 \leadstocs \{t_3 \simeq t_4\}\) for some
  \(t_3,t_4\) such that \(|t_3| =_{\beta\eta} |t_1|\) and \(|t_4| =_{\beta\eta}
  |t_2|\).
  From this, the third premise, and transitivity of convertibility, we obtain
  \(|t_3| =_{\beta\eta} |t_4|\).
  We also see from the second premise that \(\textit{FV}(t') \subseteq
  \textit{dom}(\Gamma_1,X_1:\kappa_1',\Gamma_2)\).
  So, we use the rule to conclude that \(\Gamma_1,X_1:\kappa_1',\Gamma_2 \vdash
  \beta\{t'\} \tpcheck T_4\), and the depth (1) is preserved.

  \startcase{.2cm}
  \[
    \infer{\Gamma_1,X_1:\kappa_1,\Gamma_2 \vdash \rho\ t\ @x\langle t_2 \rangle.T'\ \mhyph\ t' \tpcheck T}
    {
      \begin{array}{ccc}
        \Gamma_1,X_1:\kappa_1,\Gamma_2 \vdash t \tpsynthleads \{t_1 \simeq t_2'\}
        & \textit{FV}(t_2) \subseteq \textit{dom}(\Gamma_1,X_1:\kappa_1,\Gamma_2)
        & |t_2'| =_{\beta\eta} |t_2|
        \\ \Gamma_1,X_1:\kappa_2,\Gamma_2 \vdash [t_2/x] T' \tpsynth \star
        & \Gamma_1,X_1:\kappa_1,\Gamma_2 \vdash t' \tpcheck [t_2/x]T'
        & [t_1/x]T' \cong T
      \end{array}
    }
    % \infer{\Gamma_1,X_1:\kappa_1,\Gamma_2 \vdash \rho\ t\ @x\langle t_1 \rangle.T'\ \mhyph\ t' \tpcheck T_3}
    % {
    %   \begin{array}{ccc}
    %     \Gamma_1,X_1:\kappa_1,\Gamma_2 \vdash t \tpsynthleads \{t_1' \simeq t_2\}
    %     & \textit{FV}(t_1) \subseteq \textit{dom}(\Gamma_1,X_1:\kappa_1,\Gamma_2)
    %     & |t_1'| =_{\beta\eta} |t_1|
    %     \\ \Gamma_1,X_1:\kappa_1,\Gamma_2 \vdash [t_1/x] T' \tpsynth \star
    %     & \Gamma_1,X_1:\kappa_1,\Gamma_2 \vdash t \tpcheck [t_1/x]T'
    %     & [t_2/x]T' \cong T_3
    %   \end{array}
    % }
  \]
  By the IH on the first premise, \(\Gamma_1,X_1:\kappa_1',\Gamma_2 \vdash t'
  \tpsynthleads \{t_1 \simeq t_2'\}\).
  By the IH on the first premise of the second row,
  \(\Gamma_1,X_1:\kappa_1',\Gamma_2 \vdash [t_2/x]\ T \tpsynth \star\)
  (as \(\star\) is only convertible with itself).
  By the IH the second premise of the second row,
  \(\Gamma_1,X_1:\kappa_1',\Gamma_2 \vdash t \tpcheck [t_2/x]T'\).
  By assumption, the third premise of the second row, and transitivity of
  convertibility, we have \([t_1/x]T' \cong T_4\) (\(T_4\) is the type we must
  check the entire expression against).
  We use the rule to conclude (\(\textit{dom}(\Gamma_1,X_1:\kappa_1,\Gamma_2) =
  \textit{dom}(\Gamma_1,X_1:\kappa_1',\Gamma_2)\)), noting that the measure is
  preserved.
  
  \startcase{.2cm}
  \[
    \infer{
      \Gamma_1,X_1:\kappa_1,\Gamma_2 \vdash [ t_1 , t_2 ] \tpcheck T_3
    }{
      \begin{array}{cc}
        T_3 \leadstocs \abs{\iota}{x}{T_1}{T_1}
        & \Gamma_1,X_1:\kappa_1,\Gamma_2 \vdash t_1 \tpcheck T_1
        \\ \Gamma_1,X_1:\kappa_1,\Gamma_2 \vdash t_2 \tpcheck [t/x]^{T_1}\ T_2
        & |t_1| =_{\beta\eta} |t_2|
      \end{array}
    }
  \]
  From the first premise, assumption, and transitivity of the convertibility
  relation, \(T_4 \cong \abs{\iota}{x}{T_1}{T_2}\).
  This means that \(T_4 \leadstocs \abs{\iota}{x}{T_1'}{T_2'}\) for some
  \(T_1' \cong T_1\) and \(T_2' \cong T_2\).
  From this last congruence, we have \([t_1/x]^{T_1'}T_2' \cong [t_1/x]^{T_1} T_2\).
  By the IH on the second premise, \(\Gamma_1,X_1:\kappa_1',\Gamma_2 \vdash t_1
  \tpcheck T_1'\).
  By the IH on the first premise, second row, we have
  \(\Gamma_1,X_1:\kappa_1',\Gamma_2 \vdash t_2 \tpcheck [t/x]T_2'\).
  We use the rule to conclude, noting that the measure is preserved.
  
  % synthesis
  \startcase{.2cm}
  \[
    \infer{
      \Gamma_1,X_1:\kappa_1,\Gamma_2 \vdash t\ t' \tpsynth [t'/x]^{T'}T
    }{
      \Gamma_1,X_1:\kappa_1,\Gamma_2 \vdash t \tpsynthleads \abs{\Pi}{x}{T'}{T}
      & \Gamma_1,X_1:\kappa_1,\Gamma_2\vdash t' \tpcheck T'
    }
  \]
  By the IH on the first premise, \(\Gamma_1,X_1:\kappa_1',\Gamma_2 \vdash t
  \tpsynthleads \abs{\Pi}{x}{T'}{T}\).
  By the IH on the second premise, \(\Gamma_1,X_1:\kappa_1,\Gamma_2 \vdash t'
  \tpcheck T'\).
  By the rule, \(\Gamma_1,X_1:\kappa_1',\Gamma_2 \vdash t\ t' \tpsynth [t'/x]T\),
  and the measure is preserved.
  
  \startcase{.2cm}
  \[
    \infer{
      \Gamma_1,X_1:\kappa_1,\Gamma_2 \vdash t \cdot T' \tpsynth [T'/X]T
    }{
      \Gamma_1,X_1:\kappa_1,\Gamma_2 \vdash t \tpsynthleads
      \abs{\forall}{X}{\kappa}{T}
      & \Gamma_1,X_1:\kappa_1,\Gamma_2 \vdash T' \tpsynth \kappa'
      & \kappa'\cong\kappa
    }
  \]
  By the IH on the first premise, \(\Gamma_1,X_1:\kappa_1',\Gamma_2 \vdash t
  \tpsynthleads \abs{\forall}{X}{\kappa}{T}\).
  By the IH on the second premise, \(\Gamma_1,X_1:\kappa_1',\Gamma_2 \vdash T'
  \tpsynth \kappa''\) for some \(\kappa'' \cong \kappa'\), from which we obtain
  \(\kappa'' \cong \kappa\).
  By the rule, \(\Gamma_1,X_1:\kappa_1',\Gamma_2 \vdash t \cdot T' \tpsynth
  [T'/X]T\), and the measure is preserved.
\end{proof}

\begin{lemma}
  \label{lem:ctxt-conv-class2}
  Let \(T_1,T_2\) be types such that \(T_1 \cong T_2\).
  \begin{itemize}
  \item If \(\Gamma_1,x:T_1,\Gamma_2 \vdash \kappa\) then
    \(\Gamma_1,x:T_2,\Gamma_2 \vdash \kappa\) 
    
  \item If \(\Gamma_1,x:T_1,\Gamma_2 \vdash T \tpsynth \kappa_1\) then
    \(\Gamma_1,x:T_2,\Gamma_2 \vdash T \tpsynth \kappa_1\)

  \item If \(T_3 \cong T_4\) and \(\Gamma_1,x:T_1,\Gamma_2 \vdash t \tpcheck T_3\) then
    \(\Gamma_1,x:T_2,\Gamma_2 \vdash t \tpcheck T_4\)
    
  \item If \(\Gamma_1,x:T_1,\Gamma_2 \vdash t \tpsynth T\) then
    \(\Gamma_1,x:T_2,\Gamma_2 \vdash t \tpsynth T'\) for some \(T' \cong T\).

  \end{itemize}
  Furthermore, the resulting typing derivations have depths no larger than the assumed ones.
\end{lemma}
\begin{proof}
  By mutual induction on the assumed derivation (specifically, its depth), and
  mutually with Lemma~\ref{lem:ctxt-conv-class}.
  The measure, which omits the size of derivations of the convertibility
  relation, is used to ensure that each mutually inductive call is well-founded.
  We show the interesting cases.

  % kinding types
  \startcase{.2cm}
  \[
    \infer{
      \Gamma_1,x_1:T_1,\Gamma_2 \vdash \abs{\lambda}{x}{T}{T'} \tpsynth \abs{\Pi}{x}{T}{\kappa}
    }{
      \Gamma_1,x_1:T_1,\Gamma_2 \vdash T \tpsynth \star
      & \Gamma_1,x_1:T_1,\Gamma_2,x:T \vdash T' \tpsynth \kappa
    }
  \]
  By the IH on the first premise, \(\Gamma_1,x_1:T_2,\Gamma_2 \vdash T \tpsynth
  \star\).
  By the IH on the second premise, \(\Gamma_1,x_1:T_2,\Gamma_2,x:T \vdash T'
  \tpsynth \kappa\).
  We also have that the depths of the resulting typing derivations are no greater than
  the those used to derive the premises, so we conclude using the rule.

  \startcase{.2cm}
  \[
    \infer{
      \Gamma_1,x_1:T_1,\Gamma_2\vdash\abs{\lambda}{X}{\kappa}{T'} \tpsynth
      \abs{\Pi}{X}{\kappa}{\kappa'}
    }{
      \Gamma_1,x_1:T_1,\Gamma_2 \vdash \kappa & \Gamma_1,x_1:T_1,\Gamma_2,X:\kappa\vdash T'\tpsynth\kappa'
    }
  \]
  By the IH on the first premise, \(\Gamma_1,x_1:T_2,\Gamma_2 \vdash T \tpsynth
  \kappa\).
  By the IH on the second premise, \(\Gamma_1,x_1:T_2,\Gamma_2,X:\kappa \vdash
  T' \tpsynth \kappa'\).
  We also have that the depths of the resulting derivations are no greater than
  the those used to derive the premises, so we conclude using the rule.

  \startcase{.2cm}
  \[
    \infer{
      \Gamma_1,x_1:T_1,\Gamma_2\vdash T\ t \tpsynth [t/x]^{T'}\kappa
    }{
      \Gamma_1,x_1:T_1,\Gamma_2 \vdash T \tpsynth \abs{\Pi}{x}{T'}{\kappa}
      & \Gamma_1,x_1:T_1,\Gamma_2 \vdash t \tpcheck T'
    }
  \]
  By the IH on the first premise, \(\Gamma_1,x_1:T_2,\Gamma_2 \vdash T \tpsynth
  \abs{\Pi}{x}{T'}{\kappa}\).
  By the IH on the second premise, \(\Gamma_2,x_1:T_2,\Gamma_2 \vdash t \tpcheck
  T'\).
  We also have that the depths of the resulting derivations are no greater than
  the those used to derive the premises, so we conclude using the rule.

  \startcase{.2cm}
  \[
    \infer{
      \Gamma_1,x_1:T_1,\Gamma_2 \vdash T \cdot T' \tpsynth [T'/X]\kappa
    }{
      \Gamma_1,x_1:T_1,\Gamma_2\vdash T \tpsynth \abs{\Pi}{X}{\kappa'}{\kappa}
      & \Gamma_1,x_1:T_1,\Gamma_2 \vdash T_2 \tpsynth \kappa'' & \kappa'' \cong \kappa'
    }
  \]
  By the IH on the first premise, \(\Gamma_1,x_1:T_2,\Gamma_2 \vdash T \tpsynth
  \abs{\Pi}{X}{\kappa'}{\kappa}\).
  By the IH on the second premise, \(\Gamma_1,x_1:T_2,\Gamma_2 \vdash T_2
  \tpsynth \kappa''\).
  We also have that the depths of the resulting derivations are no greater than
  the those used to derive the premises, so we conclude using the rule, keeping
  the third premise.

  % type checking
  \startcase{.2cm}
  \[
    \infer{
      \Gamma_1,x_1:T_1,\Gamma_2 \vdash t\tpcheck T
    }{
      \Gamma_1,x_1:T_1,\Gamma_2 \vdash t\tpsynth T'
      & T' \cong T
    }
  \]
  By the IH on the first premise, \(\Gamma_1,x_1:T_2,\Gamma_2 \vdash t \tpsynth
  T''\) for some \(T'' \cong T'\).
  By transitivity of convertibility and the second premise, \(T'' \cong T\).
  By transitivity again, \(T'' \cong T_4\) (where \(T_4 \cong T\) is given to
  us).
  We also have that the depths of the resulting derivations are no greater than
  the those used to derive the premises, so we conclude using the rule.
  
  \startcase{.2cm}
  \[
    \infer{
      \Gamma_1,x_1:T_1,\Gamma_2 \vdash \absu{\lambda}{x}{t} \tpcheck T
    }{
      T \leadstocs \abs{\Pi}{x}{T_5}{T_6}
      & \Gamma_1,x_1:T_1,\Gamma_2,x:T_5\vdash t\tpcheck T_6
    }
  \]
  From the first premise, we have \(T \cong \abs{\Pi}{x}{T_5}{T_6}\), so by
  transitivity of convertibility and assumption, \(T_4 \cong
  \abs{\Pi}{x}{T_5}{T_6}\).
  This means that \(T_4 \leadstocs \abs{\Pi}{x}{T_5'}{T_6'}\) for some
  \(T_5' \cong T_5\) and \(T_6' \cong T_6\).
  By one use of the IH on the first premise, we have
  \(\Gamma_1,x_1:T_2,\Gamma_2,x:T_5 \vdash t \tpcheck T_6'\).
  By another use, we have \(\Gamma_1,x_1:T_2,\Gamma_2,x:T_5' \vdash t \tpcheck
  T_6'\) (this is well-founded, because the depth of the typing derivation this
  second use is applied to is no greater that the depth of the premise).
  We also have that the depths of the resulting derivations are no greater than
  the those used to derive the premises, so we conclude using the rule.

  \startcase{.2cm}
  \[
    \infer{
      \Gamma_1,x_1:T_1,\Gamma_2 \vdash \absu{\Lambda}{X}{t} \tpcheck T'
    }{
      T' \leadstocs \abs{\forall}{X}{\kappa}{T}
      & \Gamma_1,x_1:T_1,\Gamma_2,X:\kappa \vdash t \tpcheck T
    }
  \]
  From the first premise, \(T' \cong \abs{\forall}{X}{\kappa}{T}\).
  From transitivity of convertibility and assumption, \(T_4 \cong
  \abs{\forall}{X}{\kappa}{T}\).
  This means that \(T_4 \leadstocs \abs{\forall}{X}{\kappa''}{T''}\) for
  some \(\kappa'' \cong \kappa\) and \(T'' \cong T\).
  By mutual induction with Lemma~\ref{lem:ctxt-conv-class} on the second premise,
  \(\Gamma_1,x_1:T_1,\Gamma_2,X:\kappa'' \vdash t \tpcheck T\), and this
  derivation is no deeper than the second premise.
  By the IH, \(\Gamma_1,x_1:T_1,\Gamma_2,X:\kappa'' \vdash t \tpcheck T''\), and
  this derivation is no deeper than the second premise.
  We conclude using the rule, with the measure preserved.

  \startcase{.2cm}
  \[
    \infer{
      \Gamma_1,x_1:T_1,\Gamma_2 \vdash [ t , t' ] \tpcheck T_3
    }{
      \begin{array}{cc}
        T_3 \leadstocs \abs{\iota}{x}{T}{T'}
        & \Gamma_1,x_1:T_1,\Gamma_2 \vdash t \tpcheck T
        \\ \Gamma_1,x_1:T_1,\Gamma_2 \vdash t' \tpcheck [t/x]^{T}\ T'
        & |t| =_{\beta\eta} |t|
      \end{array}
    }
  \]
  From the first premise, \(T_3 \cong \abs{\iota}{x}{T}{T'}\).
  From transitivity of convertibility and assumption, \(T_4 \cong
  \abs{\iota}{x}{T}{T'}\).
  This means \(T_4 \leadstocs \abs{\iota}{x}{T''}{T'''}\) for some \(T'' \cong T\)
  and \(T''' \cong T'\).
  By the IH, \(\Gamma_1,x_1:T_2,\Gamma_2 \vdash t \tpcheck T''\) with depth no
  greater than the second premise (read left to right, then top to bottom).
  By the IH, \(\Gamma_1,x_1:T_2,\Gamma_2 \vdash t' \tpcheck T'''\) with depth no
  greater than the third premise.
  We conclude using the rule, keeping the fourth premise as-is and noting the
  measure is preserved.

  \startcase{.2cm}
  \[
    \infer{
      \Gamma_1,x_1:T_1,\Gamma_2 \vdash \beta\{t'\} \tpcheck T_3
    }{
      T_3 \leadstocs \{t_1 \simeq t_2\}
      \quad \textit{FV}(t')\subseteq \textit{dom}(\Gamma_1,x_1:T_1,\Gamma_2)
      \quad |t_1| =_{\beta\eta} |t_2|
    }
  \]
  From the first premise, \(T_3 \cong \{t_1 \simeq t_2\}\).
  From transitivity of the convertibility relation and assumption, \(T_3 \cong
  \{t_1 \simeq t_2\}\).
  This means \(T_4 \leadstocs \{t_3 \simeq t_4\}\) for some \(t_3,t_4\)
  such that \(|t_3| =_{\beta\eta} |t_1|\) and \(|t_4| =_{\beta\eta} |t_2|\).
  From the third premise and transitivity of convertibility, we have \(|t_3|
  =_{\beta\eta} |t_4|\).
  From the second premise, we obtain that \(\textit{FV}(t') \subseteq
  \textit{dom}(\Gamma_1,x_1:T_2,\Gamma_2)\).
  We conclude using the rule, noting that the depth remains 1.

  \startcase{.2cm}
  \[
    \infer{\Gamma_1,x_1:T_1,\Gamma_2\vdash \varphi\ t\ \mhyph\ t'\ \{t''\} \Leftrightarrow T_3}
    {
      \Gamma_1,x_1:T_1,\Gamma_2\vdash t\tpcheck \{t' \simeq t''\}
      & \Gamma_1,x_1:T_1,\Gamma_2\vdash t' \Leftrightarrow T_3
      & \textit{FV}(t'') \subseteq \textit{dom}(\Gamma_1,x_1:T_1,\Gamma_2)
    }
  \]
  By the IH on the first premise (first row), \(\Gamma_1,x_1:T_2,\Gamma_2 \vdash
  t \tpcheck \{t' \simeq t''\}\) (with no greater depth).
  If this is type synthesis, then by the second premise
  \(\Gamma_1,x_1:T_2,\Gamma_2 \vdash t' \tpsynth T_4\) for some \(T_4 \cong
  T_3\).
  If this is type checking, by the second premise \(\Gamma_1,x_1:T_2,\Gamma_2
  \vdash t' \tpcheck T_4\) since \(T_4\) (given to us in this case) is
  convertible with \(T_3\).
  Either way, we use the rule to conclude (note
  \(\textit{dom}(\Gamma_1,x_1:T_1,\Gamma_2) =
  \textit{dom}(\Gamma_1,x_1:T_2,\Gamma_2)\)), and the measure is preserved.

  \startcase{.2cm}
  \[
  \infer{\Gamma_1,x_1:T_1,\Gamma_2 \vdash \rho\ t\ @x\langle t_2 \rangle.T'\ \mhyph\ t' \tpcheck T}
  {
    \begin{array}{ccc}
      \Gamma_1,x_1:T_1,\Gamma_2 \vdash t \tpsynthleads \{t_1 \simeq t_2'\}
      & \textit{FV}(t_2) \subseteq \textit{dom}(\Gamma_1,x_1:T_1,\Gamma_2)
      & |t_2'| =_{\beta\eta} |t_2|
      \\ \Gamma_1,x_1:T_1,\Gamma_2 \vdash [t_2/x] T' \tpsynth \star
      & \Gamma_1,x_1:T_1,\Gamma_2 \vdash t' \tpcheck [t_2/x]T'
      & [t_1/x]T' \cong T
    \end{array}
  }
  \]
  By the IH on the first premise, \(\Gamma_1,x_1:T_2,\Gamma_2 \vdash t \tpsynth
  T''\) (with no greater depth) for some \(T'' \cong \{t_1 \simeq t_2'\}\).
  That means \(T'' \leadstocs \{t_3 \simeq t_4\}\) for some \(t_3,t_4\)
  such that \(|t_3| =_{\beta\eta} |t_1|\) and \(|t_4| =_{\beta\eta} |t_2'|\).
  So we further obtain that \(|t_4| =_{\beta\eta} |t_2|\).
  By the IH on the first premise of the second row, \(\Gamma_1,x_1:T_2,\Gamma_2
  \vdash [t_2/x]T \tpsynth \star\).
  By the IH on the second premise of the second row, \(\Gamma,x_1:T_2,\Gamma_2
  \vdash t \tpcheck [t_2/x]T'\).
  By assumption, the third premise of the second row, and transitivity of
  convertibility, we have \([t_1/x]T' \cong T_4\).
  We conclude with the rule, preserving the measure.

  \startcase{.2cm}
  \[
    \infer{\Gamma_1,x_1:T_1,\Gamma_2\vdash \delta\ \mhyph\ t \tpcheck T_3}
    {
      \Gamma_1,x_1 : T_1,\Gamma_2 \vdash t \tpsynth T'
      \quad T' \cong \{\absu{\lambda}{x}{\absu{\lambda}{y}{x}} \simeq \absu{\lambda}{x}{\absu{\lambda}{y}{y}}\}
    }
  \]
  By the IH on the first premise, \(\Gamma_1,x:T_2,\Gamma_2 \vdash t \tpsynth
  T''\) for some \(T'' \cong T'\).
  By transitivity of convertibility and the second premise, \(T'' \cong
  \{\absu{\lambda}{x}{\absu{\lambda}{y}{x}} \simeq
  \absu{\lambda}{x}{\absu{\lambda}{y}{y}}\}\).
  Using the \(\delta\) rule, we conclude that \(\Gamma_1,x_1:T_2,\Gamma_2 \vdash
  \delta\ \mhyph\ t \tpcheck T_4\).
  
  % type synthesis
  \startcase{.2cm}
  \[
    \infer{
      \Gamma_1,x_1:T_1,\Gamma_2 \vdash x \tpsynth T
    }{
      (x : T)\in\Gamma_1,x_1 : T_1,\Gamma_2
    }
  \]
  We have two subcases to consider.
  If \(x = x_1\), then we use the rule to produce a derivation of
  \(\Gamma_1,x_1:T_2,\Gamma_2 \vdash x \tpsynth T_2\), where by assumption \(T_1 \cong T_2\).
  Otherwise, we use the rule to produce a derivation of
  \(\Gamma_1,x_1:T_2,\Gamma_2 \vdash x \tpsynth T\).
  In both cases, the measure (depth of 1) is preserved.

  \startcase{.2cm}
  \[
    \infer{
      \Gamma_1,x_1:T_1,\Gamma_2 \vdash t\ t' \tpsynth [t'/x]^{T'}T
    }{
      \Gamma_1,x_1:T_1,\Gamma_2 \vdash t \tpsynthleads \abs{\Pi}{x}{T'}{T}
      & \Gamma_1,x_1:T_1,\Gamma_2 \vdash t' \tpcheck T'
    }
  \]
  By the IH on the premise, \(\Gamma_1,x_1:T_2,\Gamma_2 \vdash t \tpsynth T_3\)
  for some \(T_3 \cong \abs{\Pi}{x}{T'}{T}\), and the depth of this typing
  derivation is no greater than that of the first premise.
  From this we obtain that \(T_3 \leadstocs \abs{\Pi}{x}{T_4'}{T_4}\) for
  some \(T_4' \cong T'\) and \(T_4 \cong T\).
  By the IH on the last premise, \(\Gamma_1,x_1:T_2,\Gamma_2 \vdash t' \tpcheck
  T_4'\), and the depth of this derivation is no greater than that of the last premise.
  We conclude with the rule to obtain \(\Gamma_1,x_1:T_2,\Gamma_2 \vdash t\ t'
  \tpsynth [t'/x]^{T_4'}\ T_4\), with the measure preserved (\([t'/x]^{T_4'}\
  T_4\) is convertible with \([t'/x]^{T'}\ T\), since \(|\chi\ T_4'\ \mhyph\ t'|
  =_{\beta\eta} |\chi\ T'\ \mhyph\ t'|\) by erasure and \(T \cong T_4\)).

  \startcase{.2cm}
  \[
    \infer{
      \Gamma_1,x_1:T_1,\Gamma_2 \vdash t \cdot T' \tpsynth [T'/X]T
    }{
      \Gamma_1,x_1:T_1,\Gamma_2 \vdash t \tpsynthleads
      \abs{\forall}{X}{\kappa}{T}
      & \Gamma_1,x_1:T_1,\Gamma_2 \vdash T' \tpsynth \kappa'
      & \kappa'\cong\kappa
    }
  \]

  By the IH on the first premise, \(\Gamma_1,x_1:T_2,\Gamma_2 \vdash t \tpsynth
  T_3\) for some \(T_3 \cong \abs{\forall}{X}{\kappa}{T}\), with no greater
  depth.
  This means \(T_3 \leadstocs \abs{\forall}{X}{\kappa''}{T''}\) for some
  \(\kappa'' \cong \kappa\) and \(T'' \cong T\).
  By the IH on the second premise, \(\Gamma_1,x_1:T_2,\Gamma_2 \vdash T'
  \tpsynth \kappa'\), with no greater depth.
  By transitivity of convertibility, \(\kappa' \cong \kappa''\).
  We use the rule to conclude \(\Gamma_1,x_1:T_2,\Gamma_2 \vdash t \cdot T'
  \tpsynth [T'/X]T''\), with the measure preserved and the synthesized type in
  the conclusion convertible with \([T'/X]T\) (since \(T'' \cong T\)).

  \startcase{.2cm}
  \[
    \infer{
      \Gamma_1,x_1:T_1,\Gamma_2 \vdash \chi\ T\ \mhyph\ t\ \tpsynth T
    }{
      \Gamma_1,x_1:T_1,\Gamma_2 \vdash T \tpsynth \star
      & \Gamma_1,x_1:T_1,\Gamma_2 \vdash t \tpcheck T
    }
  \]
  By the IH on the first premise, \(\Gamma_1,x_1:T_2,\Gamma_2 \vdash T \tpsynth
  \star\) with no greater depth.
  By the IH on the second premise \(\Gamma_1,x_1:T_2,\Gamma_2 \vdash t \tpcheck
  T\) at no greater depth.
  We conclude using the rule, with the measure preserved.
\end{proof}

\begin{corollary}
  \label{cor:ctxt-conv-class}
  If \(\vdash \Gamma_1,X:\kappa_1,\Gamma_2\) and \(\Gamma_1 \vdash \kappa_1'\)
  with \(\kappa_1' \cong \kappa_1\) then \(\vdash \Gamma_1,X:\kappa_1',\Gamma_2\).
\end{corollary}

\begin{corollary}
  \label{cor:ctxt-conv-class2}
  If \(\vdash \Gamma_1,x_1:T_1,\Gamma_2\) and \(\Gamma_1 \vdash T_2 \tpsynth
  \star\) with \(T_1 \cong T_2\) then \(\vdash \Gamma_1,x_1:T_2,\Gamma_2\).
\end{corollary}

\begin{lemma}
  \label{lem:subst-class}
  Below, each statement separately universally quantifies over meta-variables,
  and it is assumed that typing contexts occurring in assumed derivations are
  well-formed.
  \begin{enumerate}
  \item \textbf{Kinds:}
    \begin{itemize}
    \item If \(\Gamma_1,x:T,\Gamma_2 \vdash \kappa\) and \(\Gamma_1 \vdash t
      \tpsynth T\) then \(\Gamma_1,[t/x]\Gamma_2 \vdash [t/x] \kappa\)
      
    \item If \(\Gamma_1,X:\kappa',\Gamma_2 \vdash \kappa\) and \(\Gamma_1 \vdash
      T \tpsynth \kappa'\) then \(\Gamma_1,[T/X]\Gamma_2 \vdash [T/X]\kappa\)
    \end{itemize}
    
  \item \textbf{Types}
    \begin{itemize}
    \item If \(\Gamma_1,x:T',\Gamma_2 \vdash T \tpsynth \kappa\) and \(\Gamma_1 \vdash t
      \tpsynth T'\) then \(\Gamma_1,[t/x]\Gamma_2 \vdash [t/x]T \tpsynth [t/x]\kappa\)
      
    \item If \(\Gamma_1,X:\kappa_2,\Gamma_2 \vdash T_1 \tpsynth \kappa_1\) and \(\Gamma_1
      \vdash T_2 \tpsynth \kappa_2\) then \(\Gamma_1,[T_2/X]\Gamma_2 \vdash [T_2/X]T_1 \tpsynth [T_2/X]\kappa_1\)
    \end{itemize}
    
  \item \textbf{Terms:}
    \begin{itemize}
    \item If \(\Gamma_1,x:T',\Gamma_2 \vdash t \Leftrightarrow T\) and
      \(\Gamma_1 \vdash t' \tpsynth T'\) then \(\Gamma_1,[t'/x]\Gamma_2 \vdash [t'/x]t
      \Leftrightarrow [t'/x]T\)
     
    \item If \(\Gamma_1,X:\kappa,\Gamma_2 \vdash t \Leftrightarrow T'\) and
      \(\Gamma_1 \vdash T \tpsynth \kappa\) then \(\Gamma_1,[T/X]\Gamma_2 \vdash
      [T/X]t \Leftrightarrow [T/X]T'\)
    \end{itemize}
  \end{enumerate}
\end{lemma}
\begin{proof}
  By mutual induction on the assumed derivations.
  We only show a few interesting cases, and we omit type annotations on
  substitutions when these are clear from context.
  
  \startcase{.2cm}
  \[
    \infer{\Gamma_1,X:\kappa,\Gamma_2 \vdash X_1 \tpsynth \kappa_1}{(X_1 : \kappa_1) \in \Gamma_1,X:\kappa,\Gamma_2}
  \]

  We have two cases.
  If \(X_1 = X\), then \(\kappa_1 = \kappa\) and by assumption \(\Gamma_1 \vdash
  T \tpsynth \kappa\), and the desired result holds by weakening.
  Otherwise, either \((X_1 : \kappa_1) \in \Gamma_1\) and \(X \notin
  \textit{FV}(\kappa_1)\) (which we obtain from the assumption that the typing
  context is well-formed), or else \((X_1 : \kappa_1) \in \Gamma_2\). 
  Either way, we have \(\Gamma_1,[T/X]\Gamma_2 \vdash X_1 \tpsynth
  [T/X]\kappa_1\).

  \startcase{.2cm}
  \[ \infer{\Gamma_1,x:T,\Gamma_2\vdash T_2\ t_1 \tpsynth [t_1/x_1]^{T_1}\kappa
    }{
      \Gamma_1,x:T,\Gamma_2\vdash T_2 \tpsynth \abs{\Pi}{x_1}{T_1}{\kappa}
      & \Gamma_1,x:T,\Gamma_2\vdash t_1 \tpcheck T_1
    }
  \]

  By the IH, \(\Gamma_1,[t/x]\Gamma_2 \vdash [t/x]T_2 \tpsynth
  \abs{\Pi}{x_1}{[t/x]T_1}{[t/x]\kappa}\) and \(\Gamma_1,[t/x]\Gamma_2 \vdash
  [t/x]t_1 \tpcheck [t/x]T_1\).
  By the rule, we have \(\Gamma_1,[t/x]\Gamma_2 \vdash [t/x]T_2\ [t/x]t_1
  \tpsynth [[t/x]t_1/x_1][t/x]\kappa\), where the synthesized kind is equal to
  the desired \([t/x][t_1/x_1]\kappa\).

  \startcase{.2cm}
  \[
    \infer{
      \Gamma_1,x:T,\Gamma_2\vdash \{ t_1 \simeq t_2 \} : \star
    }{
      \textit{FV}(t_1\ t_2)\subseteq\textit{dom}(\Gamma_1,x:T,\Gamma_2)
    }
  \]
  It suffices to show that \(\textit{FV}([t/x]t_1\ [t/x]t_2) \subseteq
  \textit{dom}(\Gamma_1,[t/x]\Gamma_2)\).
  It is clear \(x\) is not a free variable of this expression, that the free
  variables of \(t\) are declared in \(\Gamma_1\), and by assumption the other
  free variables of it are declared in \(\Gamma_1,[t/x]\Gamma_2\).

  \startcase{.2cm}
  \[
    \infer{\Gamma_1,x:T,\Gamma_2 \vdash x_1\tpsynth T_1}{(x_1 : T_1) \in \Gamma_1,x:T,\Gamma_2}
  \]
  We elaborate on the case where \(x_1 = x\).
  It is important that we assumed that \(\Gamma_1 \vdash t \tpsynth T\) (as
  opposed to having its type checked), as we may now replace the given rule with
  this assumed derivation without changing the definition of substitution.

  \startcase{.2cm}
  \[
    \infer{\Gamma_1,x:T,\Gamma_2\vdash \beta\{t'\} \tpcheck \{ t_1 \simeq t_2 \}}
    {\textit{FV}(t')\subseteq \textit{dom}(\Gamma_1,x:T,\Gamma_2) \quad |t_1| =_{\beta\eta} |t_2|}
  \]

  From our assumptions we may conclude \(\textit{FV}([t/x]t') \subseteq
  \textit{dom}(\Gamma_1,[t/x]\Gamma_2)\), and from the second premise that
  \(|[t/x]t_1| =_{\beta\eta} |[t/x]t_2|\).
  Thus, \(\Gamma_1,[t/x]\Gamma_2\ \vdash \beta\{[t/x]t'\}\tpcheck \{[t/x]t_1
  \simeq [t/x]t_2\}\).

  \startcase{.2cm}
  \[
    \infer{\Gamma_1,x:T,\Gamma_2 \vdash \rho\ t''\ @x_1\langle t_2 \rangle.T_2\ \mhyph\ t' \tpcheck T_1}
    {
      \begin{array}{ccc}
        \Gamma_1,x:T,\Gamma_2 \vdash t'' \tpsynthleads \{t_1 \simeq t_2'\}
        & \textit{FV}(t_2) \subseteq \textit{dom}(\Gamma_1,x:T,\Gamma_2)
        & |t_2'| =_{\beta\eta} |t_2|
        \\ \Gamma_1,x:T,\Gamma_2 \vdash [t_2/x_1] T_2 \tpsynth \star
        & \Gamma_1,x:T,\Gamma_2 \vdash t' \tpcheck [t_2/x_1]T_2
        & [t_1/x_1]T_2 \cong T_1
      \end{array}
    }
  \]
  From the IH, we have that \(\Gamma_1,[t/x]\Gamma_2 \vdash [t/x]t'' \tpsynth \{[t/x]t_1
  \simeq [t/x]t_2'\}\), that \(\Gamma_1,[t/x]\Gamma_2 \vdash [t/x][t_2/x_1]T_2
  \tpsynth \star\), and that \(\Gamma_1,[t/x]\Gamma_2 \vdash [t/x]t' \tpcheck
  [t/x][t_2/x_1]T_2\).
  From the last premise, we may conclude that \([t/x]T_1 \cong
  [t/x][t_1/x_1]T_1\).
  We also see from the premises that \(\textit{FV}([t/x]t_2) \subseteq
  \textit{dom}(\Gamma_1,[t/x]\Gamma_2)\) and that \(|[t/x]t_2'| =_{\beta\eta}
  |[t/x]t_2|\).
  Applying the rule and permuting substitutions gives us the desired result
  (note for example that \([t/x][t_2/x_1]T_1 = [[t/x]t_2/x_1][t/x]T_1\)).

  \startcase{.2cm}
  \[
    \infer{\Gamma_1,x:T,\Gamma_2 \vdash \delta\ \mhyph\ t \tpcheck T_2}
    {
      \Gamma_1,x:T,\Gamma_2 \vdash t \tpsynth T_2'
      \quad T_2' \cong \{\absu{\lambda}{x}{\absu{\lambda}{y}{x}} \simeq \absu{\lambda}{x}{\absu{\lambda}{y}{y}}\}
    }
  \]

  From the IH, we have that \(\Gamma_1[t'/x]\Gamma_2 \vdash [t'/x] t \tpsynth
  [t'/x]T_2'\).
  From the second premise, we have \([t'/x]T_2' \cong
  \{\absu{\lambda}{x}{\absu{\lambda}{y}{x} \simeq
    \absu{\lambda}{x}{\absu{\lambda}{y}{y}}}\}\).
  From the \(\delta\) rule, we conclude that
  \(\Gamma_1,[t'/x]\Gamma_2 \vdash \delta\ \mhyph\ [t'/x]t \tpcheck [t'/x]T_2\).
\end{proof}

\begin{corollary}
  \ \\
  \begin{itemize}
  \item If \(\vdash \Gamma_1,x:T,\Gamma_2\) and \(\Gamma_1 \vdash t \tpsynth T\)
    then \(\vdash \Gamma_1,[t/x]\Gamma_2\)
    
  \item If \(\vdash \Gamma_1,X:\kappa,\Gamma_2\) and \(\Gamma_1 \vdash T
    \tpsynth \kappa\) then \(\vdash \Gamma_1,[T/X]\Gamma_2\).
  \end{itemize}
\end{corollary}
\begin{proof}
  By induction on the assumed derivation, appealing to
  Lemma~\ref{lem:subst-class} at each step.
\end{proof}

\subsection{Theorem~\ref{thm:syntactic-kind-pres}}
\begin{proof}[Proof (of Theorem~\ref{thm:syntactic-kind-pres})]
  We may rule out the cases where \(T\) is a variable or formed by a type
  constructor, as this would contradict the assumption that \(T \leadstoc
  T'\) for some \(T'\).
  We omit type annotations from substitutions when they are clear from the context.

  \startcase{.2cm}
  \[
    \infer{\Gamma\vdash T\ t \tpsynth [t/x]^{T_1}\kappa_1}{\Gamma\vdash T \tpsynth
      \abs{\Pi}{x}{T_1}{\kappa_1} & \Gamma\vdash t \tpcheck T_1}
  \]

  There are two subcases to consider for the derivation of \(T\ t \leadstoc
  T'\). In the first case, we have \(T \leadstoc T''\) for some \(T''\) (so
  \(T' = T''\ t\)). By the IH, we have \(\Gamma \vdash T'' \tpsynth
  \abs{\Pi}{x}{T_1'}{\kappa_1'}\) for some \(T_1' \cong T_1\) and \(\kappa_1'
  \cong \kappa_1\) (we have from that IH that kind of \(T''\) must be
  convertible with \(\abs{\Pi}{x}{T_1}{\kappa_1}\)). By
  Lemma~\ref{lem:ctxt-conv-class2}, \(\Gamma \vdash t \tpcheck T_1'\).
  By the rule, \(\Gamma \vdash T''\ t \tpsynth [t/x]\kappa_1'\), and the
  synthesized kind is clearly convertible with \([t/x]\kappa_1\).

  In the other subcase, the subject of kinding is of the form
  \((\abs{\lambda}{x}{T_1}{T''})\ t\) and our assumed reduction is to
  \([t/x]T''\). By inversion of the kinding derivation, we have \(\Gamma,x:T_1
  \vdash T'' \tpsynth \kappa_1\). By Lemma~\ref{lem:subst-class}, we have
  \(\Gamma \vdash [t/x]T'' \tpsynth [t/x]\kappa_1\).

  \startcase{.2cm}
  \[
    \infer{\Gamma\vdash T_1 \cdot T_2 \tpsynth [T_2/X]\kappa_1}{\Gamma\vdash T_1
      \tpsynth \abs{\Pi}{X}{\kappa_2}{\kappa_1} & \Gamma\vdash T_2 \tpsynth
      \kappa_2' & \kappa_2 \cong \kappa_2'}
  \]

  There are two subcases to consider for the derivation of \(T_1 \cdot T_2
  \leadstoc T'\). In the first subcase, we have \(T_1 \leadstoc
  T_1''\) for some \(T_1''\) (so \(T' = T_1'' \cdot T_2\)). By the IH, we have
  \(\Gamma \vdash T_1'' \tpsynth \abs{\Pi}{X}{\kappa_2''}{\kappa_1''}\) for some
  \(\kappa_2'' \cong \kappa_2\) and \(\kappa_1'' \cong \kappa_1\) (we have from
  the IH that the kind of \(T_1''\) must be convertible with
  \(\abs{\Pi}{X}{\kappa_2}{\kappa_1}\)).
  By transitivity, we can conclude \(\kappa_2'' \cong \kappa_2'\).
  By the rule, we have \(\Gamma \vdash T_1'' \cdot T_2 \tpsynth
  [T_2/X]\kappa_1''\), and the synthesized kind is clearly convertible with
  \([T_2/X]\kappa_1\).

  In the other subcase, the subject of kinding is of the form
  \((\abs{\lambda}{X}{\kappa_2}{T_1''}) \cdot T_2\) and our assumed reduction is
  to \([T_2/X]T_1''\).
  By inversion of the kinding derivation, we have \(\Gamma,X:\kappa_2 \vdash
  T_1'' \tpsynth \kappa_1\).
  By Lemma~\ref{lem:ctxt-conv-class}, we have \(\Gamma,X:\kappa_2' \vdash T_1''
  \tpsynth \kappa_1'\) for some \(\kappa_1' \cong \kappa_1\), and from this and
  Lemma~\ref{lem:subst-class} we have \(\Gamma \vdash [T_2/X]T_1'' \tpsynth
  [T_2/X]\kappa_1'\).
  This is clearly convertible with \([T_2/X]\kappa_1\), as desired.
\end{proof}

\begin{corollary}
  \label{cor:syntactic-kind-pres}
  If \(\Gamma \vdash T \tpsynth \kappa\) and \(T \leadstocs T'\) then
  \(\Gamma \vdash T' \tpsynth \kappa'\) for some \(\kappa' \cong \kappa\).
\end{corollary}

\subsection{Theorem~\ref{thm:judge-valid}}

\begin{proof}[Proof (of Theorem~\ref{thm:judge-valid})]
  By induction on the assumed derivation, making implicit use of
  Corollary~\ref{cor:syntactic-kind-pres} for the shorthand \(\Gamma \vdash t
  \tpsynthleads T\) and mostly omitting annotated substitutions when these are
  clear from context (we show the first one that occurs).
  We show a few interesting cases.

  \startcase{.2cm}
  \[
    \infer{\Gamma \vdash X \tpsynth \kappa}{(X : \kappa) \in \Gamma}
  \]
  By an easy inductive argument on the assumption \(\vdash \Gamma\) and
  weakening, we obtain \(\Gamma \vdash \kappa\).

  \startcase{.2cm}
  \[
    \infer{\Gamma\vdash\abs{\lambda}{x}{T}{T'} \tpsynth \abs{\Pi}{x}{T}{\kappa}}{\Gamma \vdash T \tpsynth \star & \Gamma,x:T\vdash T'\tpsynth\kappa}
  \]
  By assumption, the first premise, and the context formation rules, we have
  \(\vdash \Gamma,x:T\).
  By the IH, we have \(\Gamma,x:T \vdash \kappa\).
  Thus, we obtain \(\Gamma \vdash \abs{\Pi}{x}{T}{\kappa}\).

  \startcase{.2cm}
  \[
    \infer{\Gamma\vdash\abs{\lambda}{X}{\kappa}{T'} \tpsynth \abs{\Pi}{X}{\kappa}{\kappa'}}{\Gamma \vdash \kappa & \Gamma,X:\kappa\vdash T'\tpsynth\kappa'}
  \]
  By assumption, the first premise, and the context formation rules, we have
  \(\vdash \Gamma,X:\kappa\).
  From this and the IH on the second premise, we have \(\Gamma,X:\kappa \vdash
  \kappa'\).
  We can then conclude that \(\Gamma \vdash \abs{\Pi}{X}{\kappa}{\kappa'}\)

  \startcase{.2cm}
  \[
    \infer{\Gamma\vdash T_1 \cdot T_2 \tpsynth [T_2/X]\kappa_1}{\Gamma\vdash T_1 \tpsynth \abs{\Pi}{X}{\kappa_2}{\kappa_1} & \Gamma\vdash T_2 \tpsynth \kappa_2' & \kappa_2 \cong \kappa_2'}
  \]
  By the IH on the first premise, we have \(\Gamma \vdash
  \abs{\Pi}{X}{\kappa_2}{\kappa_1}\), and by inversion this gives us \(\Gamma
  \vdash \kappa_2\), which yields \(\vdash \Gamma,X:\kappa_2\), and
  \(\Gamma,X:\kappa_2 \vdash \kappa_1\).
  By the IH on the second premise, \(\Gamma \vdash \kappa_2'\).
  By Lemma~\ref{lem:ctxt-conv-class} using the third premise,
  \(\Gamma,X:\kappa_2' \vdash \kappa_1\).
  By Lemma~\ref{lem:subst-class}, \(\Gamma \vdash [T_2/X]\kappa_1\).

  \startcase{.2cm}
  \[
    \infer{\Gamma\vdash T\ t \tpsynth [t/x]^{T'}\kappa}{\Gamma\vdash T \tpsynth \abs{\Pi}{x}{T'}{\kappa} & \Gamma\vdash t \tpcheck T'}
  \]
  By the IH on the first premise, \(\Gamma \vdash \abs{\Pi}{x}{T'}{\kappa}\).
  By inversion, \(\Gamma \vdash T' \tpsynth \star\) and \(\Gamma,x:T' \vdash
  \kappa\).
  So, from the first of these and by context formation \(\vdash \Gamma,x:T'\).
  From the second premise and the \(\chi\) rule, \(\Gamma \vdash \chi\ T\
  \mhyph\ t \tpsynth T\).
  Finally, by Lemma~\ref{lem:subst-class}, \(\Gamma \vdash [t/x]^{T} \kappa\).

  \startcase{.2cm}
  \[
    \infer{\Gamma\vdash t\ t' \tpsynth [t'/x]^{T'}T}{\Gamma\vdash t \tpsynthleads \abs{\Pi}{x}{T'}{T} & \Gamma\vdash t' \tpcheck T'}
  \]
  By the IH on the first premise, \(\Gamma \vdash \abs{\Pi}{x}{T'}{T} \tpsynth
  \star\).
  By inversion of this, we obtain both \(\Gamma \vdash T' \tpsynth \star\), and
  that \(\Gamma,x:T' \vdash T \tpsynth \star\).
  From the first of these and the context formation rules, we have \(\vdash
  \Gamma,x:T'\).
  Finally, from Lemma~\ref{lem:subst-class} we have \(\Gamma \vdash [t'/x]T
  \tpsynth \star\).

  \startcase{.2cm}
  \[
    \infer{
      \Gamma\vdash t \cdot T' \tpsynth [T'/X]T
    }{
      \Gamma\vdash t \tpsynthleads \abs{\forall}{X}{\kappa}{T} & \Gamma\vdash T'
      \tpsynth \kappa' & \kappa'\cong\kappa
    }
  \]
  From the IH on the first premise, \(\Gamma \vdash \abs{\forall}{X}{\kappa}{T}
  \tpsynth \star\).
  By inversion of this, we have both \(\Gamma \vdash \kappa\) and
  \(\Gamma,X:\kappa \vdash T \tpsynth \star\).
  From the first of these, context formation rules, and assumption we have
  \(\vdash \Gamma,X:\kappa\).
  By the IH on the second premise, we also have \(\Gamma \vdash \kappa'\), and
  from this we have \(\vdash \Gamma,X:\kappa'\).
  We combine this with the earlier derivation of \(\Gamma,X:\kappa \vdash T
  \tpsynth \star\), the third premise, and Lemma~\ref{lem:ctxt-conv-class} to
  obtain \(\Gamma,X:\kappa' \vdash T \tpsynth \star\).
  Finally, we use Lemma~\ref{lem:subst-class} to obtain \(\Gamma \vdash [T'/X]T
  \tpsynth \star\).

  % Type checking, not needed
  % \startcase{.2cm}
  % \[
  %   \infer{\Gamma \vdash \rho\ t\ @x\langle t_1 \rangle.T'\ \mhyph\ t' \tpcheck T}
  %   {
  %     \begin{array}{ccc}
  %       \Gamma \vdash t \tpsynthleads \{t_1' \simeq t_2\}
  %       & \textit{FV}(t_1) \subseteq \textit{dom}(\Gamma)
  %       & |t_1'| =_{\beta\eta} |t_1|
  %       \\ \Gamma \vdash [t_1/x] T' \tpsynth \star
  %       & \Gamma \vdash t' \tpcheck [t_1/x]T'
  %       & [t_2/x]T' \cong T
  %     \end{array}
  %   }
  % \]
  % The desired result is given by the first premise of the second row.

  \startcase{.2cm}
  \[
    \infer{\Gamma \vdash \chi\ T\ \mhyph\ t\ \tpsynth T}
    {\Gamma \vdash T \tpsynth \star & \Gamma \vdash t \tpcheck T}
  \]
  Give to us by the first premise.
\end{proof}

\end{document}